\documentclass[11pt,letterpaper]{article}

\usepackage{amsmath,amsthm,amsfonts,amssymb}
\usepackage{thm-restate}
\usepackage{fullpage}
\usepackage[utf8]{inputenc}
\usepackage[dvipsnames]{xcolor}
\usepackage{xspace}
\usepackage[shortlabels]{enumitem}
\usepackage[hypertexnames=false,colorlinks=true,urlcolor=Blue,citecolor=Green,linkcolor=BrickRed]{hyperref}
\usepackage[capitalise]{cleveref}
\usepackage[OT4]{fontenc}
\usepackage[ruled]{algorithm}
\usepackage{ifpdf}
\usepackage{todonotes}
\usepackage{thmtools}
\usepackage{mathtools}
\usepackage[noend]{algpseudocode}
\usepackage{textpos}
\usepackage{makecell}
\usepackage{multirow}

\usepackage[margin=1in]{geometry}

\usepackage{crossreftools}
\newcommand{\cqed}{\ensuremath{\lhd}}
\makeatletter
\newenvironment{claimproof}{\par
	\pushQED{\cqed}%
	\normalfont \topsep6\p@\@plus6\p@\relax
	\trivlist
	\item\relax
	{\itshape
		Proof of the claim\@addpunct{.}}\hspace\labelsep\ignorespaces
}{%
	\hfill\popQED\endtrivlist\@endpefalse
}
\makeatother

\theoremstyle{plain}
\newtheorem{theorem}{Theorem}[section]
\newtheorem{lemma}[theorem]{Lemma}

\newtheorem{observation}[theorem]{Observation}
\newtheorem{Definition}[theorem]{Definition}

\newtheorem{claim}[theorem]{Claim}

\newtheorem{assumption}{Assumption}

\def\poly{\operatorname{poly}}
\def\polylog{\operatorname{polylog}}

\newcommand{\email}[1]{\href{mailto:#1}{#1}}

\title{Strongly Polynomial Parallel Work-Depth Tradeoffs\\ for Directed SSSP\\
}

\date{}
\author{Adam Karczmarz\thanks{University of Warsaw, Poland. \email{a.karczmarz@mimuw.edu.pl}. Supported by the National Science Centre (NCN) grant no. 2022/47/D/ST6/02184.}
\and Wojciech Nadara\thanks{University of Warsaw, Poland. \email{w.nadara@mimuw.edu.pl}. Supported by European Union's Horizon 2020 research and innovation programme, grant agreement No. 948057 — BOBR.}
\and Marek Sokołowski\thanks{Max Planck Institute for Informatics, Saarland Informatics Campus, Germany. \email{msokolow@mpi-inf.mpg.de}.}}

\newcommand{\Oh}{\ensuremath{O}}
\newcommand{\Ot}{\ensuremath{\widetilde{O}}}
\newcommand{\dist}{\mathrm{dist}}

\newcommand{\wei}{w}
\newcommand{\floor}[1]{\left\lfloor #1 \right\rfloor}
\newcommand{\ceil}[1]{\left\lceil #1 \right\rceil}

\newcommand{\M}{\mathbb{M}}
\newcommand{\N}{\mathbb{N}}

\newcommand{\R}{\mathbb{R}}
\newcommand{\Rn}{\R_{\geq 0}}

\newcommand{\Yc}{\mathcal{Y}}

\newcommand{\known}{S}
\newcommand{\near}[2]{N_{{#1}}({#2})}
\newcommand{\NL}{\mathrm{NL}}
\newcommand{\indeg}{\mathrm{indeg}}
\newcommand{\outdeg}{\mathrm{outdeg}}
\newcommand{\iso}{\mathrm{iso}}
\newcommand{\diff}{\mathrm{diff}}

\begin{document}

\maketitle

\iftrue
\begin{textblock}{20}(-1.8, 8.7)
	\includegraphics[width=40px]{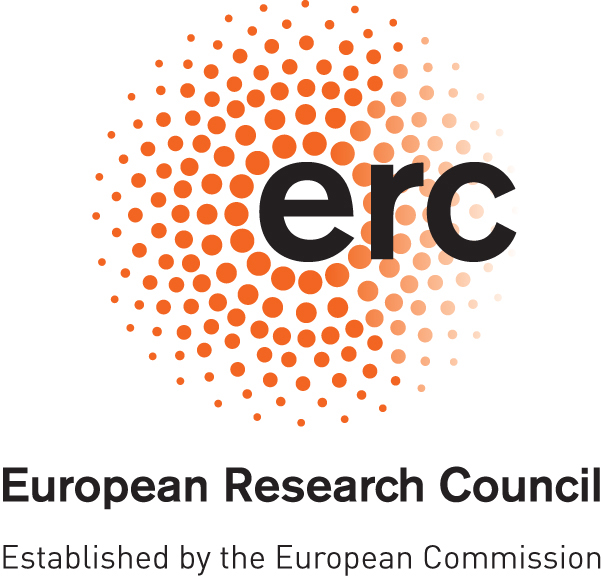}%
\end{textblock}
\begin{textblock}{20}(-2.05, 9.0)
	\includegraphics[width=60px]{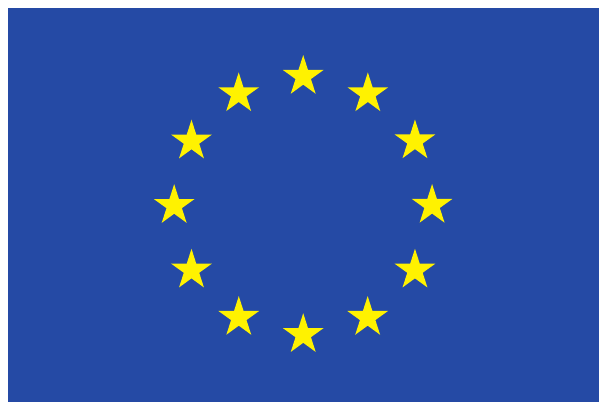}%
\end{textblock}
\fi

\begin{abstract}
  In this paper, we show new strongly polynomial work-depth tradeoffs for computing single-source shortest paths (SSSP) in non-negatively weighted directed graphs in parallel.
  Most importantly, we prove that directed SSSP can be solved within $\Ot(m+n^{2-\epsilon})$ work and $\Ot(n^{1-\epsilon})$ depth
  for some positive $\epsilon>0$.
  In particular, for dense graphs with non-negative real weights, we provide the first nearly work-efficient strongly polynomial algorithm with sublinear depth.
  
  Our result immediately yields improved strongly polynomial parallel algorithms for min-cost flow and the assignment problem.
  It also leads to the first non-trivial strongly polynomial dynamic algorithm for minimum mean cycle.
  Moreover, we develop efficient parallel algorithms in the Word RAM model for several variants of SSSP in graphs with exponentially large edge weights.
\end{abstract}

\section{Introduction}

Single-source shortest paths (SSSP) is one of the most fundamental and well-studied graph problems. 
The famous Dijkstra's algorithm~\cite{dijkstra1959note} solves the problem on directed graphs ${G=(V,E)}$ with~$n$ vertices, $m$ edges, and non-negative real edge weights
in ${O(m+n\log{n})}$ time if an appropriate priority queue is used~\cite{FredmanT87}.
Dijkstra's algorithm, in its classical form, finds an exact solution by operating on the input weights via additions and comparisons only.
It is also \emph{strongly polynomial}, i.e., its running time on real data is polynomial in $n,m$ exclusively\footnote{And does not abuse the arithmetic model, e.g., runs in polynomial space for rational inputs on a Turing machine.}.
In particular, it does not depend on the input weights' magnitude, which is the case for \emph{weakly polynomial} algorithms.

Devising improved strongly polynomial shortest paths algorithms has historically proved very challenging.
Only very recently, the decades-old state-of-the-art strongly polynomial SSSP bounds \cite{bellman1958routing,FredmanT87} have witnessed first breakthroughs: for non-negative weights~\cite{DuanMMSY25} (in the sparse case), and for possibly-negative weights~\cite{Fineman24}.
This is despite numerous previous improvements in the weakly polynomial regime (for integer data), e.g.,~\cite{BernsteinNW22, BoasKZ77, ChenKLPGS22, CohenMSV17, GabowT89, Goldberg95, Hagerup00, Thorup00, Thorup04} or in the undirected case~\cite{DuanMSY23}.

\paragraph{Parallel SSSP.} In this paper, our focus is on strongly polynomial \emph{parallel} SSSP algorithms for non-negatively weighted \emph{directed} graphs (digraphs).
The current de facto standard for studying the efficiency of parallel shared-memory (PRAM) algorithms is the \emph{work-depth} model (e.g.,~\cite{jaja1992introduction,blelloch1996programming}).
The work of a parallel algorithm is its sequential running time, whereas the depth is the longest chain of sequential dependencies between the performed operations.
In this paper, when quantifying work and depth, we always stick to the $\Ot(\cdot)$ notation\footnote{$\Ot$ suppresses $\polylog(n)$ factors. That is, whenever we write $\Ot(f(n))$, we mean $O(f(n)\polylog{n})$. In particular, $\Ot(1)$ captures $O(\log^2{n})$ and $\Ot(n/t)$ captures $O((n/t)\polylog{n})$ for all parameters $t\in[1,n]$.} as we are concerned about polynomial improvements anyway.\footnote{Moreover, this way, we do not need to specify the precise shared memory access model (e.g., EREW, CRCW) as polylogarithmic-overhead reductions between the variants are known~\cite{jaja1992introduction}.}
A parallel algorithm is called \emph{nearly work-efficient} if its work matches, up to polylogarithmic factors, the best-known sequential bound for the studied problem.

Dijkstra's algorithm can be parallelized to run in $\Ot(m)$ work and $\Ot(n)$ depth if one uses parallel priority queues~\cite{PaigeK85, BrodalTZ98}.
It is also well-known that all-pairs shortest paths (APSP) -- and thus also SSSP -- can be computed within $\Ot(n^3)$ work and $\Ot(1)$ depth via repeated squaring-based exponentiation of the weighted adjacency matrix using the $(\min,+)$ product.
Identifying feasible work-depth pairs or tradeoffs for exact SSSP beyond these classical ones has been the subject of a~large body of prior work~\cite{Blelloch0ST16, BringmannHK17, CaoF23, Cohen97, ForsterN18, GazitM88, KleinS93, KleinS97, RozhonHMGZ23, ShiS99, Spencer97, UllmanY91}; however, their majority studied weakly polynomial bounds, and mostly in integer-weighted graphs.
Notably,~\cite{RozhonHMGZ23} showed a general $\Ot(1)$-overhead reduction to $(1+\epsilon)$-approximate SSSP for digraphs with polynomially-bounded non-negative integer weights, which enabled lifting known approximate SSSP bounds (e.g.~\cite{CaoFR20, CaoFR20a, Zwick02}) to the exact setting.
In particular, this implied a nearly work-efficient algorithm
with $n^{1/2+o(1)}$ depth, a result shown independently by~\cite{CaoF23}.
Spencer~\cite{Spencer97} showed an $\Ot(m+nt^2)$ work and $\Ot(n/t)$ depth trade-off (for any $t\in [1,n]$) that does not require integrality, but merely $\poly(n)$-bounded weights.
\Cref{tab:bounds} lists the state-of-the-art bounds (wrt.\ $\Ot$) for exact parallel SSSP under various restrictions.

\paragraph{Strongly polynomial parallel SSSP.} Parallel SSSP in the strongly polynomial regime appears much less understood.
A simplified version of Spencer's SSSP tradeoff~\cite{Spencer97} (not described explicitly in the literature, to the best of our knowledge; see~\Cref{sec:basic}) yields $\Ot(n^2t)$ work and $\Ot(n/t)$ depth for any parameter $t\in [1,n]$.
Note that this constitutes a natural interpolation between parallel Dijkstra and repeated-squaring for dense graphs.
For sparse graphs,~\cite{BringmannHK17} show a better tradeoff with $\Ot(mn+t^3)$ work and $\Ot(n/t)$ depth that even works for negative weights.

Strikingly, for directed SSSP with non-negative weights, none of the known strongly polynomial parallel algorithms for directed graphs with non-trivial, truly sublinear depth achieve subquadratic work.
In particular, none attain sub-Bellman--Ford $O(mn^{1-\epsilon})$ work for all densities $m\in [n,n^2]$, and none are nearly-work efficient for \emph{any} density, such as~$m\approx n^2$.
It is thus natural to ask:
\begin{quote}
  \emph{Is there a strongly polynomial SSSP algorithm for non-negatively weighted directed graphs with
  $\Ot(m+n^{2-\epsilon})$ work and $\Ot(n^{1-\epsilon})$ depth, for some $\epsilon>0$?}
\end{quote}
Note that in the \emph{weakly polynomial} regime, algorithms with such work/depth characteristics have been known for subexponential real weights~\cite{Spencer91,Spencer97}
for more than 30 years.
Interestingly, the posed question has also been known to have an affirmative answer for the special case of \emph{undirected} graphs~\cite{ShiS99}.
Both Spencer~\cite{Spencer97} and Shi and Spencer~\cite{ShiS99} noted that their methods fail to produce non-trivial bounds
for directed graphs in the strongly polynomial regime.
They explicitly asked for strongly polynomial parallel SSSP algorithms for digraphs
with comparable bounds.

It is worth noting that, in the strongly polynomial regime, directed SSSP computation (with non-negative weights) is the bottleneck
of many state-of-the-art optimization algorithms, e.g.,~\cite{EdmondsK72, HuangJQ25, OlverV20, Orlin93}.
Therefore,~improving the parallel SSSP bounds can immediately lead to improved parallel algorithms for those.

\subsection{Our results}
As our main result, we give an affirmative answer to the posed question and show the following \emph{deterministic} strongly polynomial work-depth tradeoff for directed SSSP.
\begin{theorem}
  \label{thm:main}
  Let $G = (V, E)$ be a~weighted digraph with weights in $\R_{\ge 0}$.
  Single-source shortest paths in $G$ can be computed within $\Ot(m + n^{9/5}t^{17/5})$ work and $\Ot(n / t)$ depth for any $t \in [1, n^{1/17}]$.
\end{theorem}
For dense enough graphs with $m=\Omega(n^{1+{4/5}+\epsilon})$, the algorithm behind~\Cref{thm:main} can run in near-optimal $\Ot(m)$ work and truly sublinear in $n$ depth.
This is the first strongly polynomial~result to improve upon parallel Dijkstra's depth while remaining nearly work-efficient for some graph~density. 
For sparse graphs, using a significantly simpler algorithm,
we achieve a slightly better trade-off.
\begin{restatable}{theorem}{sparsetradeoff}\label{thm:sparse-tradeoff}
  Let $G=(V,E)$ be a weighted digraph with weights in $\Rn$.
  Single-source shortest paths in $G$ can be computed within
  $\Ot(m^{5/3}t^2+m^{3/2}t^{7/2})$ work and $\Ot(m/t)$ depth for $t\in [1,m^{1/2}]$.
\end{restatable}

\begin{table}[ht!]
  \centering
  \begin{tabular}{|c|c|c|c|}
  \hline
  \textbf{Algorithm} & \textbf{Work} & \textbf{Depth} & \textbf{Restrictions} \\ \hline
    \multicolumn{4}{|c|}{\emph{strongly polynomial}}\\ \hline
    parallel Dijkstra \cite{PaigeK85} & $\Ot(m)$ & $\Ot(n)$ & --- \\ \hline
    \makecell{repeated squaring wrt.\\ min-plus product~\cite{Williams18}} & $n^{3-o(1)}$ & $\Ot(1)$ & --- \\ \hline
    \makecell{$t$-nearest vertices-based\\\Cref{thm:simple-tradeoff} (simplified~\cite{Spencer97})} & $\Ot(n^2t)$ & $\Ot(n/t)$ & --- \\ \hline
    Bringmann et al.~\cite{BringmannHK17} & $\Ot(mn+t^3)$ & $\Ot(n/t)$ & --- \\ \hline
    \multirow{2}{*}{Shi and Spencer \cite{ShiS99}} & $\Ot(m + nt^2)$ & $\Ot(n/t)$ & \multirow{2}{*}{undirected} \\ \cline{2-3}
   & $\Ot(mn/t + t^3)$ & $\Ot(n/t)$ & \\ \hline
    \multirow{2}{*}{\textbf{Ours}} & $\Ot(m^{5/3}t^2+m^{3/2}t^{7/2})$ & $\Ot(m/t)$ & $t \le m^{1/2}$ \\ \cline{2-4}
     & $\Ot(m + n^{9/5}t^{17/5})$ & $\Ot(n/t)$ & $t \le n^{1/17}$ \\ \hline
    \multicolumn{4}{|c|}{\emph{weakly polynomial}}\\ \hline
    Spencer \cite{Spencer97} & $\Ot((m + nt^2)\log{L})$ & $\Ot((n/t)\log{L})$ & --- \\ \hline
    Klein and Subramanian~\cite{KleinS93} & $\Ot(m^2\log{L})$ & $\Ot(\log{L})$ & integer \\ \hline
    Cao and Fineman \cite{CaoF23} & $\Ot(m\log{L})$ & $n^{1/2 + o(1)}\log{L}$ & integer \\ \hline
   \makecell{~\cite{RozhonHMGZ23}\\+ approximate SSSP \cite{CaoFR20a}} & $\Ot((mt^2 + nt^4)\log{L})$ & $(n^{1/2 + o(1)}/t)\log{L}$ & \makecell{integer\\$t\leq n^{1/2}$} \\ \hline
    \makecell{~\cite{RozhonHMGZ23}\\+ approximate APSP \cite{Zwick02}} & $\Ot(n^\omega \log{L})$ & $\Ot(\log{L})$ & integer \\ \hline
    \cite{BrandGJV25} & $\Ot(m\log{L}+n^{1.5}\log^2{L})$ & $\Ot(n^{1/2}\log{L})$ & integer \\ \hline
\end{tabular}
  \caption{State-of-the-art weakly- and strongly polynomial exact SSSP bounds for general graphs (with $m\geq n$) and the additional assumptions they require.
  The table only contains results that~have constituted strict improvements wrt $\Ot(\cdot)$ in some regime for some graph density
  and have stood since.
  The tradeoff parameter $t$, unless restricted further, is any integer in $[1,n]$ to be chosen freely.
  $L$ denotes the maximum edge weight; for~\cite{Spencer97} we assume the minimum positive weight equals~$1$.
  }\label{tab:bounds}
\end{table}

\subsubsection{Applications}\label{sec:intro-applications}

\paragraph{Strongly polynomial min-cost flow.} \Cref{thm:main} has immediate applications in strongly polynomial min-cost flow computation.
The state-of-the-art (sequential) algorithm~\cite{Orlin93} (see also a variant in~\cite[Section~10.7]{network-flows}) in this regime runs in $\Ot(m^2)$ time.
Viewed as a parallel algorithm, it runs within $\Ot(m\cdot W_{\mathrm{SSSP}}(n,m))$ work and $\Ot(m\cdot D_{\mathrm{SSSP}}(n,m))$ depth (see~\cite[p. 339]{Orlin93}), given a strongly polynomial (non-negative) SSSP algorithm running within $W_{\mathrm{SSSP}}(n,m)$ work and $D_{\mathrm{SSSP}}(n,m)$ depth.
By using~\Cref{thm:main} for SSSP computation therein, we obtain an~algorithm with $\Ot(m^2 + mn^{9/5}t^{17/5})$ work and $\Ot(mn/t)$ depth for $t \in [1, n^{1/17}]$.
For sufficiently dense graphs, this yields an~improved depth bound without increasing work (ignoring polylogarithmic factors).

An analogous improvement in depth is obtained for the assignment problem, i.e., minimum-cost (perfect) matching.
The best known strongly polynomial algorithm for the assignment problem is the \emph{successive shortest paths algorithm}
(see~\cite{EdmondsK72} or~\cite[Section~12.4]{network-flows}) that runs in $\Ot(nm)$ time, which amounts to $O(n)$ SSSP computations.
Apart from running SSSP $O(n)$ times, computing the optimal assignment requires additional $O(n^2)$ work and $\Ot(n)$ depth.
Hence again, \Cref{thm:main} implies a~strongly polynomial parallel algorithm for the assignment problem with $\Ot(mn + n^{14/5}t^{17/5})$  work and $\Ot(n^2/t)$ depth -- improving depth without increasing work for sufficiently dense graphs.

\paragraph{Lexicographical optimization, exponential weights.} We argue that our strongly polynomial parallel SSSP algorithm might
prove superior to the best weakly polynomial algorithms
even in a realistic word RAM model where only arithmetic operations are performed in $O(1)$ time.

Suppose $G=(V,E)$ is a digraph with edges assigned unique labels from $\{1,\ldots,m\}$.
Let us define the \emph{lex-bottleneck weight} of a~path in $G$ to be the sequence of weights of edges on the path, sorted non-increasingly.
The \emph{lex-bottleneck shortest path} between two vertices of the graph is the path with the lexicographically minimum lex-bottleneck weight.
That is, it is the path minimizing the maximum weight of an~edge, breaking ties by the second-maximum weight, followed by the third-maximum weight etc.
Consider computing single-source lex-bottleneck paths in parallel.

Observe that the lex-bottleneck paths problem can be easily cast as SSSP problem if the edge labeled $j$ is assigned an integer weight $2^j$.
This way, we can reach out for fast weakly-polynomial parallel SSSP algorithms (such as~\cite{CaoF23, RozhonHMGZ23}) for integer weights to solve the problem.
However, even ignoring arithmetic cost, their work is $\Omega(m\log{L})$ and their depth is $\Omega(\log{L})$, where $L$ is the largest weight.
$L=\Theta(2^m)$ in our case, so we get $\Omega(m^2)$ work- and $\Omega(m)$ depth bounds.

On the other hand, our SSSP algorithm works in the comparison-addition model, so it could run within the same or slightly worse bounds on the (parallel) word RAM if only we could support all the performed operations on edge/path weights in low work and depth.
We indeed show that this is possible when using weights of the form $2^x$ (where $x\in [0,m]$ is an integer), even when these weights are not necessarily distinct.
In such a case, we still achieve $\Ot(m+n^{2-\epsilon})$ work and $\Ot(n^{1-\epsilon})$ depth on the parallel word RAM.
These developments are discussed in~\Cref{sec:large-weights}.

\paragraph{Dynamic minimum mean cycle and min-ratio cycle.} Finally, \Cref{thm:main} might be useful combined with parametric search~\cite{Megiddo83} to yield new strongly polynomial bounds in the \emph{sequential} setting as well.
As an example, we show the first non-trivial strongly polynomial \emph{dynamic} algorithm maintaining the minimum mean cycle.
The minimum mean cycle problem is a classical graph problem with applications in, e.g., min-cost-flow computation~\cite{Tardos85} and Markov Decision Processes (e.g.,~\cite{Madani02}).
It can be solved in strongly polynomial $O(nm)$ time~\cite{Karp78}.

We prove that, using~\Cref{thm:main}, we can devise a procedure updating the minimum mean cycle in $\Ot(mn^{1-1/22})$ worst-case time after an edge insertion,
i.e., polynomially faster than recompute-from-scratch. To the best of our knowledge, before our work, no non-trivial dynamic minimum mean cycle algorithms, or even applications of parametric search in the dynamic setting, have been described so far.
In fact, our dynamic algorithm applies even to a more general minimum cost-to-time ratio (min-ratio) cycle problem~\cite{DantzigBR67, Lawler1972},
whose static state-of-the-art strongly polynomial bound is currently $\Ot(m^{3/4}n^{3/2})$~\cite{BringmannHK17}.
See~\Cref{sec:mean-cycle}.
We believe that our parallel SSSP algorithm may find more applications in combination with parametric search in the future.

\subsection{Technical summary}
Our starting point is a variant of the general strategy used by~\cite{Spencer97}: to simulate $t$ consecutive steps
of Dijkstra's algorithm using a variant of low-depth repeated squaring.
Repeated squaring can identify the $t$ closest vertices $\near{t}{u}$ from each $u\in V$ in $G$ in $\Ot(m+nt^2)$ work and $\Ot(1)$ depth~\cite{ShiS99, Cohen97}.
Note that $\near{t}{s}$ contains precisely the $t$ first vertices that Dijkstra's algorithm with source $s$ would discover.
To allow reusing this idea for identifying subsequent groups of $t$ remaining vertices nearest from the source $s$,
it is enough to ``contract'' $\near{t}{s}$ into $s$ and adjust the graph's weights properly.
Unfortunately, the sets $\near{t}{u}$ computed in one iteration might not be very useful in later iterations if they overlap with $\near{t}{s}$ significantly.
Thus, the basic strategy~(\Cref{sec:basic}) simply repeats $t$ closest vertices computation
$\lceil n/t\rceil$ times; this requires at least $\Theta(n^2t)$ work.

Roughly speaking, the main technical development of~\cite{Spencer97} is a certain decremental data structure for maintaining\footnote{Actually, the \emph{nearby list} data structure of~\cite{Spencer97} maintains much more relaxed objects and only supports deletions with a certain structure imposed, but this is strong enough for the studied applications.} the sets $\near{t}{u}$ subject to vertex deletions from $G$.
This way, the sets of closest $t$ vertices do not need to be recomputed from scratch when vertices are deemed discovered.
Sadly, the amortized analysis in~\cite{Spencer97} is inherently weakly polynomial; it charges costly $t$-nearest vertices recomputations for each $u$
to a factor-$2$ increase in the radius of a ball containing $\near{t}{u}$.

We use a different approach. 
Our main idea is to apply a more depth-costly preprocessing computing static nearest-neighbors-like data (the \emph{near-lists}) that:
\begin{enumerate}[(1)]
  \item remains useful through $\poly(n)$ steps of discovering $t$ closest remaining vertices, and
  \item allows locating $\near{t}{s}$ by running the low-depth algorithm of~\cite{ShiS99, Cohen97} on a smaller graph~$H$ that preserves distances from $s$ to its nearby vertices.
\end{enumerate}
The overall strategy in the preprocessing is to identify a sublinear set of \emph{heavy} vertices such that the remaining \emph{light} vertices do not appear in too many precomputed lists of close vertices.
In a way\footnote{Whenever we discover new vertices in our SSSP algorithm, the subsequent discovery steps should only use shortest paths avoiding these vertices. So, in a way, the graph we need to repeatedly explore is decremental, i.e., it evolves by vertex deletions. In the fully dynamic APSP algorithms, such as~\cite{AbrahamCK17, GutenbergW20b, Mao24a, Thorup05},  handling the vertex-decremental case is the main challenge as well. In~\cite{GutenbergW20b} (and other dynamic shortest paths papers, e.g.,~\cite{BrandK23}), the heavy-light vertex distinction is applied to guarantee that a single vertex deletion can only break a bounded number of preprocessed shortest paths.}, this is reminiscent of an idea used in the fully dynamic APSP algorithm of~\cite{GutenbergW20b}.
Implementing this strategy in low depth turns out highly non-trivial, though, and we cannot rely on the sequential process in~\cite{GutenbergW20b}.
We end up computing near-lists using $n$ parallel Dijkstra~runs on the graph~$G$ undergoing vertex deletions, synchronized after each vertex visit.
For good congestion guarantees, such a preprocessing crucially requires the in-degrees of~$G$ to be sufficiently low, e.g., constant.

The resulting algorithm (\Cref{sec:sparse}) can already achieve subquadratic work and sublinear depth for sparse enough graphs.
However, it does not yield any interesting bounds if the average degree of $G$ is high enough,
let alone if $G$ is dense.
To deal with that, we devise an additional global data structure that filters the vertices' neighborhoods and keeps both
in-degrees sufficiently small and out-degrees $O(t)$ during the described preprocessing.
Again, to this end, the data structure introduces another subset of $o(n)$ \emph{permanently heavy} vertices
that require special treatment and needs, as we prove, only $\Ot(n/t)$ low-depth updates while the original SSSP computation proceeds.

Preprocessing the obtained filtered graph $G_0\subseteq G$ with $O(nt)$ edges yields weaker properties~of the produced near-lists, though (compared to preprocessing the entire graph $G$).
Luckily, as we prove, the issue can be addressed with sufficiently low polynomial overhead: An extra low-depth improvement step can restore the guarantees we could get if we had preprocessed $G$ in the first~place.
\subsection{Further related work}
Parallel SSSP has also been studied for negative weights in the weakly polynomial regime~\cite{AshvinkumarBCGH24, CaoFR22, FischerHLRS25} and the best known bounds~\cite{AshvinkumarBCGH24, FischerHLRS25} match those for the non-negative case up to polylogarithmic factors.
Approximate SSSP has~also gained much attention in the parallel setting~\cite{AndoniSZ20, CaoFR20, Cohen00, ElkinN19, ElkinM21, Li20, RozhonGHZL22} and a near-optimal deterministic algorithm has been shown in the undirected case~\cite{RozhonGHZL22}.
\subsection{Organization}
The rest of this paper is organized as follows.
We start by setting up notation and simplifying assumptions in~\Cref{sec:preliminaries}.
In~\Cref{sec:basic}, we describe the $\Ot(n^2t)$-work $\Ot(n/t)$-depth tradeoff for directed SSSP that constitutes the base of our main results.
In~\Cref{sec:sparse}, we develop our novel subquadratic-work sublinear-depth trade-off in the sparse case (see~\Cref{thm:sparse-tradeoff}). 
The framework presented there is then extended in~\Cref{sec:dense} to arbitrary densities, thus proving~\Cref{thm:main}.
In Sections~\ref{sec:large-weights}~and~\ref{sec:mean-cycle} we discuss the applications mentioned in~\Cref{sec:intro-applications}.

\section{Preliminaries}\label{sec:preliminaries}
In this paper, we deal with non-negatively real-weighted digraphs $G=(V,E)$.
We use $n$ and $m$~to denote the numbers of vertices and edges of $G$, respectively.
W.l.o.g.\ we assume $m\geq n$.
We write $e=uv$~to denote a directed edge from its tail $u$ to its head $v$. We denote by $\wei(e)$ the weight of edge~$e$.
By $N^{+}_G(v)$ and $N^{-}_G(v)$ we denote the out- and in-neighborhoods, respectively, of a vertex in $G$.
We also use the standard notation $\outdeg_G(v)=|N^{+}_G(v)|$ and $\indeg_G(v)=|N^{-}_G(v)|$.

We define a path $P=s\to t$ in $G$ to be a sequence of edges $e_1\ldots e_k$, where $u_iv_i=e_i\in E$, such that $v_i=u_{i+1}$, $u_1=s$, $v_k=t$.
If $s=t$, then $k=0$ is allowed.
Paths need not be simple, i.e., 
they may have vertices or edges repeated.
We often view paths as subgraphs and write~$P\subseteq G$.

The weight (or length) $\wei(P)$ of a path $P$ equals the sum $\sum_{i=1}^k\wei(e_i)$.
We denote by $\dist_G(s,t)$ the weight of the shortest $s\to t$ path in $G$. 
If the graph in question is clear from context, we sometimes omit the subscript and write $\dist(s,t)$ instead.
\subsection{Simplifying assumptions}\label{sec:simplifying}
We make the following additional simplifying assumptions, somewhat similar to those in~\cite{DuanMMSY25}.
\begin{assumption}\label{ass:prefix}
For any path $P$ in $G$ and its proper subpath $P'\subseteq P$, $\wei(P')<\wei(P)$.
\end{assumption}
\begin{assumption}\label{ass:path-weights}
For any fixed $u\in V$, no two paths $P_1=u\to x$ and $P_2=u\to y$ in $G$, where $x\neq y$, have the same weight.
\end{assumption}
In~\Cref{sec:assumptions}, we argue that these assumptions can be met in the comparison-addition model.
Although ensuring them might require using $O(1)$-size vector-weights instead of scalar~weights, for brevity, we often neglect that and pretend they still come from $\Rn$.
E.g.,~we may write $0$ to denote the weight of an empty path and assume the edge weights are simply strictly positive real numbers.

In the following, we also focus only on computing the distances from a single source in $G$, and not the actual paths or an SSSP tree.
This is due to the fact that equipped with single-source distances and~\Cref{ass:prefix} (i.e., assuming positive weights), an SSSP tree can be found by picking for each $v$ (in parallel) an arbitrary \emph{tight} edge $e=uv$ satisfying $\dist(s,u)+\wei(uv)=\dist(s,v)$.

\section{The basic trade-off for directed SSSP}\label{sec:basic}
In this section, we present the basic trade-off for non-negative SSSP in directed graphs.
At a high level, this is a simplification of Spencer's SSSP algorithm~\cite{Spencer97} that one would arrive at after depriving it of the efficient but complex data structures exploiting the bounded-weights assumption.

Let $G=(V,E)$ be a weighted digraph with weights in $\Rn$ and a source vertex $s\in V$.
When computing distances from the source $s$, w.l.o.g.\ we can assume that $s$ has no incoming edges.
Moreover, we can w.l.o.g.\ assume that $G$ has no parallel edges.

Let $t\in[1,n]$ be an integer.
For simplicity, assume w.l.o.g.\ that $n=t\alpha+1$ for some integer~$\alpha\geq 1$.

The general strategy behind all SSSP algorithms presented in this paper is as follows. Similarly to Dijkstra's algorithm, we maintain a set $\known\subseteq V$ of vertices whose distance from~$s$ has already been established.
Initially, we set $\known=\{s\}$.
While $\known\neq V$, we repeatedly perform a \emph{discovery step} identifying the next $t$ vertices from $V\setminus\known$ that are closest to $s$ in small depth.
If we perform each such step in $\Ot(1)$ depth, the overall depth will be $\Ot(\alpha)=\Ot(n/t)$.

The remainder of this section is devoted to showing a basic implementation of the above strategy. 
Let $\near{t}{u}\subseteq V$ be the set of $t$ nearest vertices from $u$ wrt.\ $\dist(u,\cdot)$ and distinct from $u$.
If fewer than~$t$ vertices are reachable from $u$ in $G$, then we define $\near{t}{u}$ to contain all these vertices.
Note that $\near{t}{u}$ is defined unambiguously by~\Cref{ass:path-weights}.

To proceed, we need the following folklore generalization of the repeated squaring-based APSP algorithm.
Its variants have been previously used in the context of SSSP in both directed~\cite{Spencer97} and undirected graphs~\cite{Cohen97, ShiS99}.
The lemma is essentially proved~in~\cite[Lemma~1]{ShiS99}.
For completeness, we include its proof sketch in~\Cref{sec:nearest-neighbors}.
\begin{restatable}{lemma}{nearestneighbors}\label{lem:nearest-neighbors}
  Let $t\in [1,n]$.
  Using $\Ot(nt^2+m)$ work and $\Ot(1)$ depth, one can compute for each $u\in V$
  the set $\near{t}{u}$ along with the corresponding distances from $u$.
\end{restatable}
Note that \Cref{lem:nearest-neighbors} allows performing the first discovery step of the general strategy, i.e., identifying the~$t$ closest vertices from $s$ in $G$ within $\Ot(nt^2+m)$ work and $\Ot(1)$ depth.
Afterwards, we have ${\known=\{s\}\cup\near{t}{s}}$.
Since in the next step we need to identify the $t$ \emph{next} closest vertices from~$s$, i.e., skipping those we have already found, the computed sets $\near{t}{u}$ might be of no use anymore. For example, it's easy to construct an example where all the sets $\near{t}{u}$ are equal.

Nevertheless, we would still like to reapply~\Cref{lem:nearest-neighbors} for finding the $t$ nearest vertices in $V\setminus\known$ from $s$ in a black-box way by setting up a graph $G'$ on $V'\subseteq V\setminus\known\cup \{s\}$ where the sought vertices correspond to the $t$ nearest vertices from $s$ in $G'$.
One way to achieve that is to contract $\{s\}\cup \near{t}{s}$ into the new source vertex $s$, and shift the corresponding weights of edges previously originating in $\near{t}{s}$ accordingly.
The following definition formalizes the contraction that we use.
\begin{Definition}\label{def:contract}
  Let $X\subseteq V\setminus \{s\}$. Let ${G_0'=(V\setminus X,E')}$ where $E'$ contains:
  \begin{itemize}
    \item for each edge $uv\in E$ with $u,v\in V\setminus X$, the edge $uv$ with the same weight $\wei(uv)$ as in $G$,
    \item for each edge $xv\in E$ with $x\in X$ and $v\in V\setminus X$, the edge $sv$ with weight $\dist_G(s,x)+\wei(xv)$.
  \end{itemize}
  We say that $G'$ is obtained by \emph{contracting $X$ into $s$}, if $G'$ is obtained from $G_0'$ by discarding parallel edges except those with the smallest weights.
\end{Definition}
The following lemma shows the key properties of the contracted graph.
\begin{lemma}\label{lem:contract}
  Let $G'$ be obtained from $G$ by contracting $\near{t}{s}$ into $s$ as in~\Cref{def:contract}.
  Then:
  \begin{itemize}
    \item for any $v\in V\setminus \near{t}{s}$ we have $\dist_G(s,v)=\dist_{G'}(s,v)$.
    \item for all $u,v\in V\setminus \near{t}{s}$ we have $\dist_G(u,v)\leq \dist_{G'}(u,v)$.
  \end{itemize}
\end{lemma}
\begin{proof}
  First of all, it is easy to verify that $G'$ satisfies Assumptions~\ref{ass:prefix}~and~\ref{ass:path-weights}
  if these are ensured in~$G$ as described in~\Cref{sec:simplifying}.

  Consider the former item.
  For any shortest path $P=s\to v$ in $G$, there is a corresponding path $P'=s\to v$ in $G'$ of the same weight.
  Hence, $\dist_{G'}(s,v)\leq \dist_G(s,v)$.
  On the other hand, by construction, a path $P'=s\to v$ in $G'$ does not pass through $s$ as an intermediate vertex.
  Hence, there is a corresponding path in $G$ of weight no larger than $w(P')$ obtained by mapping the initial edge of $P'$ to a path in $G$.
  As a result, $\dist_{G}(s,v)\leq \dist_{G'}(s,v)$.

  Now consider the latter item. Let $P'$ be the shortest $u\to v$ path in $G'$.
  $P'$ does not go through~$s$, since $s$ has no incoming edges in either $G$ or $G'$.
  Consequently, every edge of $P'$ is an edge of $G$ with the same weight.
  Hence, $P'\subseteq G$ and we obtain $\dist_{G'}(u,v)\geq \dist_G(u,v)$, as desired.
\end{proof}

\noindent By~\Cref{lem:contract}, we can specify the general strategy as follows.
While $V\neq \{s\}$:
\begin{enumerate}[(a)]
  \item compute $\near{t}{s}$,
  \item record the distances to these vertices,
  \item contract $\near{t}{s}$ into $s$.
\end{enumerate}
Above, we assume that in item~(c), the original graph $G$ is replaced by the contracted graph; in particular, $V$ shrinks into $V\setminus \near{t}{s}$.
This way, we obtain an SSSP instance of the same kind. 
Note that the vertex set is shrunk in every step by $t$ vertices, so the number of discovery steps is $O(n/t)$.

Let us now explain how the adjacency lists are maintained under the contractions into~$s$ in this strategy.
Whenever some $\near{t}{s}$ is contracted, for all outgoing edges $xy$ of $\near{t}{s}$ in parallel, we add $sy$ to the adjacency list of $s$ with updated weight unless $y\in \{s\}\cup \near{t}{s}$ or $s$
already had an outgoing edge $sy$.
In the latter case, we might decrease the weight of $sy$ if $\dist_G(sx)+\wei(xy)<\wei(sy)$.
Observe that this way, every edge of the initial $G$ is updated at most once in the process.
The contraction can be performed in near-optimal work and $\Ot(1)$ depth if the adjacency lists are stored as parallel dictionaries, e.g., parallel balanced BSTs with edges keyed by their heads.
Therefore, maintaining the adjacency lists of $G$ under the performed contractions can be implemented in near-optimal work and depth.
As a result, the time bottleneck of the computation lies in computing the nearest vertices.
If we use~\Cref{lem:nearest-neighbors} for that, the depth is near-linear in the number of discovery steps, i.e., $\Ot(n/t)$, whereas the total work is $\Ot((n/t)\cdot (nt^2+m))=\Ot(n^2t+mn/t)$.

A closer inspection into the algorithm behind~\Cref{lem:nearest-neighbors} shows that the $\Ot(m)$ term in the work bound in~\Cref{lem:nearest-neighbors} could be avoided if the input graph was provided with adjacency lists with all but the top-$t$ outgoing edges according to the weight filtered out.
Note that we can perform such filtering and maintain it subject to contractions within $\Ot(m)$ total work and $\Ot(n/t)$ total depth if we assume that the adjacency lists are additionally stored as parallel balanced BSTs sorted by edge weights and supporting near-linear work and $\Ot(1)$-depth batch insertions and deletions (see, e.g.~\cite{BlellochFS16, PaulVW83}).
This observation reduces the required work to $\Ot(n^2t + m) = \Ot(n^2t)$.

We summarize the described basic algorithm with the following theorem.
\begin{theorem}[{based on~\cite{Spencer97}}]\label{thm:simple-tradeoff}
  Let $t\in [1,n]$.
  Let $G=(V,E)$ be a weighted digraph with weights in $\Rn$.
  Single-source distances in $G$ can be computed within $\Ot(n^2t)$ work and $\Ot(n/t)$ depth.
\end{theorem}

\section{An improved trade-off for sparse graphs}\label{sec:sparse}
In this section, we show that the general strategy from~\Cref{sec:basic} can be implemented more efficiently for sparse enough graphs.

Recall that the basic implementation from~\Cref{sec:basic} recomputes, upon each discovery step, the nearest $t$ vertices from all $v\in V$ on \emph{entire} $G$ (with previously discovered vertices contracted into~$s$) in low depth using~\Cref{lem:nearest-neighbors}.
Moreover, only the set $\near{t}{s}$ is used in a single discovery step, even though~\Cref{lem:nearest-neighbors} yields the sets $\near{t}{v}$ for all $v\in V$ as a by-product; these are disregarded.

At a very high level, we will address the two downsides.
We first propose an additional preprocessing that computes nearest-neighbors-like data that is useful through multiple discovery steps.
Then, we describe a more sophisticated implementation of a discovery step computing $\near{t}{s}$.
This algorithm runs the low-depth algorithm of~\Cref{lem:nearest-neighbors} on a smaller auxiliary graph $H$ devised from the preprocessed data, that preserves distances from $s$ to its nearby vertices.

In the following, we assume that every vertex \emph{except possibly the source} $s$ has degree $O(1)$~and thus $m=O(n)$.
This assumption can be dropped in a standard way by vertex splitting~(see e.g., \cite[Section~2]{DuanMMSY25}); in general, the reduction increases the number of vertices to $\Theta(m)$,~though. 

We will also often use the notation $G-X$, for some $X\subseteq V$, to denote the graph $G[V\setminus X]$, i.e.,~$G$ with
vertices $X$ removed.

The algorithm operates in \emph{phases} of $\ell\leq n/t$ discovery steps; we will pick $\ell$ later.
Every one of the $\lceil \frac{n-1}{\ell\cdot t}\rceil$ phases starts with a costly preprocessing that will help perform the next $\ell$ steps.

\subsection{Preprocessing}\label{sec:sparse-prep}
The goal of the preprocessing is to compute analogues of $\near{t}{v}$ that are robust against vertex deletions.
Specifically, the preprocessing delivers:
\begin{enumerate}
  \item A subset $Z\subseteq V$ of \emph{heavy} vertices (including $s$) that require special treatment.
    
  \item For each $u\in V$, a \emph{near-list} $\NL(u)$ \emph{containing $u$ and at most $t$ other vertices from $V$}.
    \linebreak Every vertex $v$ stored in $\NL(u)$ is accompanied with a \emph{distance}.
    Formally, $\NL(u)\subseteq V\times \Rn$, but we will sometimes
    write $v\in \NL(u)$ if $v$ is stored with some distance in $\NL(u)$.
    Note that the distances stored in $\NL(u)$ are distinct by~\Cref{ass:path-weights}.
\end{enumerate}
Assume w.l.o.g.\ that $|V|\geq t+1$.
We call the vertices $V\setminus Z$ \emph{light}.

Let us now comment on the intuition behind heavy vertices and the obtained near-lists.
Recall that, in the worst case, in the basic strategy, it could happen that the $t$ discovered vertices cover all the sets $\near{t}{v}$, $v\in V$, making these sets useless in the following discovery step.
To circumvent this problem, the close vertices are identified gradually while the processing proceeds.
Vertices close to many $v\in V$ at once are detected early on and excluded from further consideration,
preventing any particular vertex from appearing in the near-lists too frequently.
Heavy vertices~$Z$ are precisely the said excluded vertices, whose ``congestion'' reached some chosen threshold.
As one would expect, excluding the subset $Z$ from consideration requires a more subtle definition of what a near-list $\NL(u)$ is;  it is impossible to guarantee that $\NL(u)$ still contains $\near{t}{u}$.
Instead, one can ensure that $\NL(u)$ contains, roughly speaking, some vertices that are closer in $G$ than the $t$ closest vertices in $G-Z$, i.e., closest via paths avoiding the heavy vertices $Z$.

We now formalize the above intuition. Let $p\in [2,n]$ be another integer parameter.
We want the heavy vertices and the near-lists to have together the following properties:
\begin{enumerate}[(i)]
  \item\label{prop:small-congest} Each vertex $u\in V$ appears in $O(p)$ near-lists.
  \item\label{prop:heavy} Each heavy vertex $z\in Z$ appears in $\Theta(p)$ near-lists.
  \item\label{prop:path-length} If $(v,d)\in \NL(u)$ (that is, $d$ is the distance accompanying $v$ in $\NL(u)$), then $d$ equals the weight of some $u\to v$ path in $G$.
  \item\label{prop:dominate2} Fix $u\in V\setminus Z$. 
    If $\NL(u)=\{(v_0,d_0),\ldots,(v_k,d_k)\}$, and $0=d_0< d_1< \ldots< d_k$, then:
    \begin{enumerate}[(a)]
      \item \label{prop:sparse-a} If a~vertex $v$ is reachable from $u$ in $G-Z$ and $v\notin \NL(u)$, then $k=t$ and ${d_t< \dist_{G-Z}(u,v)}$.
      \item \label{prop:sparse-b} $d_i\leq \dist_{G-Z}(u,v_i)$ holds for all $i=0,\ldots,k$.
    \end{enumerate}

\end{enumerate}
The remainder of this section is devoted to proving the following.
\begin{lemma}\label{lem:sparse-preprocessing}
  Let $p\in [2,n]$ be an integer and assume all vertices of $G$ except $s$ have degree $O(1)$.
  A subset $Z$ of heavy vertices of size $O(nt/p)$, along with a collection of near-lists $\NL(u)$, $u\in V$,
  satisfying properties~\ref{prop:small-congest}-\ref{prop:dominate2}, can be 
  computed within $\Ot(nt)$ work and $\Ot(t)$ depth.
\end{lemma}
Initially, we put $Z=\{s\}$ and the set $Z$ might grow in time.
For each $u\in V$ in parallel, we initiate and run a truncated-Dijkstra-like computation whose visited vertices (along with the established distances) will form $\NL(u)$.
The Dijkstra-like procedure will be truncated after visiting~$t$ vertices or running out of vertices to visit.
Initially, $\NL(u)=\{(u,0)\}$ for all $u\in V$.

As in Dijkstra's algorithm, for each $u$ we maintain a priority queue $Q_u$ storing the candidate vertices that can be reached from the already visited vertices using a single edge.
However, we require such an edge to connect \emph{light} vertices.
Formally, if $\NL(u)=\{(v_0,d_0),\ldots,(v_j,d_j)\}$ at some point, then $Q_u$ stores the vertices of $V\setminus (Z \cup \NL(u))$ that can be reached via a single edge from $\NL(u)\setminus Z$.
The key of $y\in Q_u$ equals
\begin{equation}\label{eq:queue-key}
  \min\left\{d_i+\wei(v_iy) : v_iy\in E, v_i,y\notin Z, 0\leq i\leq j\right\}.
\end{equation}

In a single \emph{iteration} (numbered $1,\ldots,t$), if $Q_u\neq\emptyset$, then the procedure extracts the minimum-keyed vertex $x$ from $Q_u$ and adds it to $\NL(u)$ with the corresponding distance equal to its key in~$Q_u$.
Once $x$ is inserted into $\NL(u)$, it stays there forever.
Inserting $x$ to $\NL(u)$ might require adding (or decreasing the keys of) some out-neighbors of~$x$ outside $\NL(u)\cup Z$ to $Q_x$ (this is also required upon initialization, before iteration $1$).
If $Q_u=\emptyset$, the iteration does nothing.

Importantly, the $n$ individual Dijkstra runs are synchronized after each run performs a single iteration.
Hence, the synchronization happens at most~$t$ times, after each iteration.
Synchronization is also the only time when the set of heavy vertices may grow.
Specifically, all vertices from $V\setminus Z$ that currently appear in at least $p$ near-lists are declared heavy and enter $Z$.
Note that afterwards, as a result of the move, some of the priority queues $Q_u$ might contain either heavy vertices
or light vertices~$y$ reached optimally via an edge whose tail has now become heavy (i.e., whose key in $Q_u$ does not equal~\eqref{eq:queue-key} anymore).
To maintain the queues $Q_u$ and keys therein efficiently subject to insertions to $Z$,
we additionally:
\begin{itemize}
  \item Store the queues as parallel balanced BSTs. Then, they can be filtered out (in parallel) of the newly established
heavy vertices within near-optimal work and $\Ot(1)$ depth.
  \item For each $u\in V$ and $y\in Q_u$, store the values $d_i+\wei(v_iy)$, where $v_iy\in E$, $i\leq j$, $v_i\notin Z$,
    in a parallel balanced BST as well. Recall that their minimum~\eqref{eq:queue-key} constitutes the key of $y$ in~$Q_u$.
\end{itemize}
Finally, the Dijkstra runs rooted at the newly identified heavy vertices are terminated as we will not need
their near-lists anyway.
\subsubsection{Correctness} We now prove that, after the described preprocessing completes, the desired requirements~\ref{prop:small-congest}-\ref{prop:dominate2} are met.
First of all, note that, by construction, each $\NL(u)$ contains the pair $(u,0)$ and at most $t$ other vertices of $G$.

\paragraph{\ref{prop:small-congest}~and~\ref{prop:heavy}.} By construction, every light vertex appears in less than $p$ near-lists.
Similarly, every heavy vertex except $s$ appears in $\geq p$ near-lists, since no vertices are ever removed from the near-lists.

To see that a vertex $z\in Z$ appears in $O(p)$ near-lists, consider the synchronization that establishes $z$ as heavy, after the $i$-th iteration.
During the $i$-th iteration, $z$ is light and thus is a member of $<p$ near-lists.
For $z$ to become the $i$-th near-list member of $\NL(u)$, there had to be a then-light vertex $v_j\in \NL(u)$, $j<i$, such that $v_jz\in E$.
Let us charge such an appearance of $z$ in $\NL(u)$ to the edge $v_jz\in E$.
To bound the number of possible vertices $u$, it is enough to bound the total charge of the incoming edges $vz\in E$, where $v$ was light.
For a fixed edge $vz$, the lightness of~$v$ guarantees it appears in $<p$ near-lists; hence $vz$ is charged $<p$ times.
The number of possible vertices~$v$ is no more than $\indeg(z)$.
Therefore, the sought total charge of the incoming edges of $z$ is at most
\begin{equation}
  (p-1)\cdot \indeg(z).
\end{equation}
Recall that $\indeg(z)=O(1)$ for $z\neq s$ by our sparsity assumption.
We conclude that upon the $i$-th synchronization, $z$ appears in $(p-1)+(p-1)\cdot O(1)=O(p)$ near-lists.
This completes the proof of properties~\ref{prop:small-congest}~and~\ref{prop:heavy}.

From property~\ref{prop:heavy} it follows that $|Z|=O(nt/p)$ as the total size of near-lists constructed~is~$O(nt)$.

\paragraph{\ref{prop:path-length}.} This property can be proved by induction on the time when some $v$ enters $\NL(u)$.
Initially, we have $\NL(u)=\{(u,0)\}$.
Clearly, the claim holds for the initial element of $\NL(u)$.
When a vertex~$y$ is added to $\NL(u)$, its corresponding distance is $d_i+w(v_iy)$ (see~\eqref{eq:queue-key}) for some $v_i$ that has been added to $\NL(u)$ before $u$.
By the inductive assumption, $d_i$ is the length of some $u\to v_i$ path in $G$.
Thus, $d_i+w(v_iy)$ is the length of a $u\to y$ path in $G$ as well.

\paragraph{\ref{prop:dominate2}.}
Finally, to establish property~\ref{prop:dominate2}, we prove it for a fixed $u\in V\setminus Z$ in a more general version, where~$t$ denotes the iteration performed most recently and $Z$ denotes the heavy vertices after the synchronization following that iteration and the corresponding $Z$-update.
Let us consider the initialization to be the $0$-th iteration.

  Observe that the subsequent distances established in the Dijkstra runs are non-decreasing (as in the standard version of Dijkstra's algorithm), and in fact strictly increasing (by~\Cref{ass:path-weights}).
  As a result, if we view the near-lists as sorted by~$d_i$, their new elements can be thought of as being added only at their ends.

  We prove property~\ref{prop:dominate2} by induction on $t$.
  Consider $t=0$. Both items~(a)~and~(b) follow by~\Cref{ass:prefix} since the $u,v$-distance for $u\neq v$ is always larger than the weight $0$ of an empty path, and for $u=v$ it equals the weight of an empty path.

  Now suppose the lemma holds for $t\geq 0$. We prove that it also holds for $t':=t+1$.
  In the following, let $Z,Z'$ be the heavy sets after iterations $t$ and $t+1$, respectively.
  Similarly, let $\NL(u)$ and $\NL'(u)$ be the respective near-lists of $u$ after iterations $t$ and $t+1$.
  By construction, we have $Z\subseteq Z'$ and $\NL(u)\subseteq \NL'(u)$.
  
  Consider a vertex $v$ reachable from $u\in V\setminus Z'$ in $G-Z'$ such that $v\notin \NL'(u)$.
  The vertex~$v$ is reachable from $u$ in $G-Z$ since $Z\subseteq Z'$ and $v\notin \NL(u)$, so $\NL(u)$ contains $t+1$ vertices $u=v_0,\ldots,v_t$ by item~(a) of the inductive assumption.
  
  Let $P$ be the shortest $u\to v$ path in $G-Z$. 
  Let $y$ be the earliest vertex on $P$ such that $y\notin \NL(u)$.
  As $u\in \NL(u)$ (this is guaranteed in initialization), the vertex $x$ preceding $y$ on~$P$ exists and is a member of $\NL(u)$, i.e., $v_j=x$ for some $j\in \{0,\ldots,t\}$.
  Therefore, upon iteration $t+1$, $y\in Q_u$ holds and the key of $y$ in the queue is at most $d_j+\wei(xy)$.
  As $Q_u$ is non-empty upon iteration $t+1$, a new element $v_{t+1}$ joins the near-list of $u$ as a result of that iteration.
  This proves that indeed the size of $\NL'(u)$ is $t+2$, as desired.
  
  By item~(b) of the inductive assumption, during iteration $t+1$, the key $\gamma$ of $y$ in $Q_u$ satisfies: \[\gamma\leq d_j+\wei(xy)\leq \dist_{G-Z}(u,v_j)+\wei(xy)=\dist_{G-Z}(u,x)+\wei(xy)=\dist_{G-Z}(u,y).\]
  Hence, if it is not~$y$ that is selected as $v_{t+1}$, then since $d_{t+1}$ is the weight of some $u\to v_{t+1}$ path in~$G$, by~\Cref{ass:path-weights} we conclude:
  \[ d_{t+1}<\gamma\leq \dist_{G-Z}(u,y)\leq \dist_{G-Z}(u,v).\]
  Otherwise, if $y$ is selected as $v_{t+1}$, then $y\neq v$ by $v\notin \NL'(u)$ and by~\Cref{ass:prefix} we have
  $\dist_{G-Z}(u,y)<\dist_{G-Z}(u,v)$, so we obtain
  \[ d_{t+1}=\gamma\leq \dist_{G-Z}(u,y)< \dist_{G-Z}(u,v).\]
  As $\dist_{G-Z}(u,v)\leq \dist_{G-Z'}(u,v)$, in both cases we conclude $d_{t+1}<\dist_{G-Z'}(u,v)$ which establishes item~(a), as desired.

  To show item~(b), first consider a shortest path $Q=u\to v_{t+1}$ in $G-Z$.
  Let $v_j,y$ be some two consecutive vertices on $Q$ such that $v_j\in \NL(u)$ and $y\notin \NL(u)$.
  Note that $v_{t+1}$ equals $y$ or appears later than $y$ on $Q$, so by~\Cref{ass:prefix}, $\dist_{G-Z}(u,v_{t+1})\geq \dist_{G-Z}(u,y)$.
  Thus, by $Z\subseteq Z'$ and item~(b) of the inductive assumption, we have:
  \[ \dist_{G-Z'}(u,v_{t+1})\geq \dist_{G-Z}(u,v_{t+1})\geq \dist_{G-Z}(u,y)=\dist_{G-Z}(u,v_j)+w(v_jy)\geq d_j+w(v_jy).\]
  So the key of $y$ in $Q_u$ during iteration $t+1$ is at most $d_j+w(v_jy)\leq \dist_{G-Z'}(u,v_{t+1})$.
  The extracted key $d_{t+1}$ is no larger than that, which proves the desired inequality $d_{t+1}\leq \dist_{G-Z'}(u,v_{t+1})$.
  Clearly, the inequality $d_i\leq \dist_{G-Z'}(u,v_i)$ for $i\in \{0,\ldots,t\}$ holds by $Z\subseteq Z'$ and the inductive assumption.

  \subsubsection{Running time}
  Each of the $n$ parallel Dijkstra runs visits at most $t+1$ vertices $v$ other than $s$ (recall that ${s\in Z}$ initially).
  Processing a visited vertex $v$ requires updating the keys of $O(\outdeg(v))$ vertices
  and subsequently reinserting that many vertices with updated keys into the relevant queues.
  These updates can be performed in parallel for all outgoing edges using a parallel balanced BST,
  and thus processing a vertex costs $\Ot(\outdeg(v))$ work and $\Ot(1)$ depth.
  Since out-degrees except $\outdeg(s)$ are $O(1)$ by our sparsity assumption, both total work and depth per Dijkstra run is $\Ot(t)$.

  Identifying new heavy vertices upon synchronization, and maintaining the queues and their keys to account for that
  can be realized in $\Ot(1)$ depth if we store counters of the numbers of appearances of individual vertices $v\in V$
  in the near-lists constructed so far.
  The additional work needed for that can be charged to the operations performed by the Dijkstra runs, as introducing new heavy vertices can only remove elements that the Dijkstra runs inserted into the relevant parallel trees.

  We conclude that the total work used by the preprocessing is $\Ot(nt)$, and the total depth is~$\Ot(t)$.

  \subsection{Discovery step}
  We now describe the implementation of a single discovery step that makes use of the preprocessed near-lists.
  Recall from~\Cref{sec:basic} that in a single discovery step, our goal is to compute $\near{t}{s}$ in the ``current'' digraph $G$ with all the previously discovered vertices contracted into the current source~$s$.
  The algorithm in~\Cref{sec:basic} applied~\Cref{lem:nearest-neighbors} to the entire graph $G$.
  Our main idea here is that it is enough to apply~\Cref{lem:nearest-neighbors} to a certain carefully chosen subgraph $H$ of the current $G$.

  The preprocessed data is meant to be shared among $\ell$ consecutive discovery steps forming a single phase.
  Thus, before we proceed, we need to introduce some more notation to easily refer to the state of the graph at the time when the preprocessing happened.

  Fix some phase and a moment of time after some discovery steps have been performed.
  Suppose the preprocessing at the beginning of the phase was run on the graph ${G_0=(V_0,E_0)}$ and produced a heavy set $Z\subseteq V_0$ and a collection of near-lists $\NL(u)$ for $u\in V_0$.
  Let~$U\subseteq V_0$ be the set of vertices that have been discovered since the last preprocessing happened.
  That is, the current graph $G=(V,E)$ equals $G_0$ with the vertices $U$ contracted into the source $s$.
  In particular, $V\cap U=\emptyset$: recall that every vertex contracted into $s$ is removed from $V$.
  Since a phase consists of $\ell$ discovery steps, at any time we have $|U|\leq \ell\cdot t$ and $U=\emptyset$ holds at the beginning of the phase.

  For a vertex $v\in V$, let us denote by $E_t(v)$ its $t$ (or all, if $\outdeg(v)<t$) outgoing edges in $G$ with smallest weights\footnote{Assuming super-constant $t$, $E_t(v)$ contains all the outgoing edges of $v$ for all $v\neq s$ (by the constant-degree assumption we use in this section).
  Here, applying the distinction between $E_t(v)$ and all outgoing edges of $v$ might appear excessive and artificial.
  However, it will become much more important later, in our main development (\Cref{sec:dense}).
  Foreseeing that, in the current section, we use slightly more general definitions and constructions than needed.}.
  We will also sometimes refer to $E_t(v)$ as the \emph{top-$t$} edges of $v$, and to the heads of those edges as the \emph{top-$t$} out-neighbors of $v$.
  
  Now, let us define the following two important subsets of $V$ that require special treatment:
  \begin{itemize}
    \item $Z^*=Z\setminus U$ -- the remaining (i.e., undiscovered) heavy vertices of $G$. Recall that $s\in Z^*$.
    \item $B=\{u\in V : \NL(u)\cap U\neq\emptyset\}$ -- the vertices whose near-lists (that, recall, were initialized for the graph $G_0$) contain a vertex discovered since the start of the phase.
      The vertices $B$ may be viewed as \emph{bad} as their near-lists are no longer guaranteed to be of full size, which was the case when the phase started.
  \end{itemize}
  Additionally, we define the following auxiliary local extension of the special vertices $B\cup Z^*$:
  \[ Y=\bigcup_{b\in B\cup Z^*}\bigcup_{bv\in E_t(b), v\in V\setminus (B\cup Z^*)} \NL(v).\]
  The set $Y$ is obtained by unioning, for each of the top-$t$ out-neighbors $v$ of every vertex $b\in B\cup Z^*$, the near-list $\NL(v)$, unless $v$ is heavy or bad.
    In particular, for every such $v$, the near-list $\NL(v)$ included in $Y$ is guaranteed to be free of vertices discovered in the current phase.
    
    Our main claim in this section is that, when looking for the $t$ closest vertices $\near{t}{s}$, we can limit ourselves to a significantly smaller subgraph of $G$ spanned by $Z^*\cup B\cup Y$. Formally, we prove:
  \begin{lemma}\label{lem:equivalent-subgraph}
    Let $G_t=(V,\bigcup_{v\in V}E_t(v))$, that is, $G$ with all but top-$t$ outgoing edges of every $v\in V$ removed.
    Let $H=G_t[Z^*\cup B\cup Y]$. Let $u\in \near{t}{s}$. Then $u\in V(H)$ and $\dist_G(s,u)=\dist_H(s,u)$.
  \end{lemma}
  \begin{proof}
    Let $P$ be the shortest $s\to u$ path in $G$. 
    Recall that $V\cap U=\emptyset$.
    We will show that $P\subseteq G_t$ and $V(P)\subseteq V(H)$.
    This will prove $u\in V(H)$ and $P\subseteq H$, and thus $\dist_H(s,u)\leq \dist_G(s,u)$.
    The inequality $\dist_G(s,u)\leq \dist_H(s,u)$ is trivial by $H\subseteq G$.

    First, consider any edge $ab$ of $P$. For contradiction, suppose $ab\notin E_t(a)$. Then, we have $|E_t(a)|=t$.
    By~\Cref{ass:path-weights}, for any $ac\in E_t(a)$, we have $\wei(ac)<\wei(ab)$.
    Thus, \[\dist_G(s,c)\leq \dist_G(s,a)+\wei(ac)<\dist_G(s,a)+\wei(ab)=\dist_G(s,b)\leq \dist_G(s,u).\]
    Since the number of distinct vertices $c\neq s$ is $t$, it follows that there are at least $t$ vertices $\neq s$ closer from $s$ than $u$.
    This contradicts $u\in \near{t}{s}$.
    Therefore, $P\subseteq G_t$ indeed holds.

    Now, we prove $V(P)\subseteq Z^*\cup B\cup Y=V(H)$. Consider some vertex $v\in V(P)\setminus (B\cup Z^*)$.
    Since $s\in Z^*$, there exists some latest vertex $x\in V(P)\cap (B\cup Z^*)$ appearing before $v$ on $P$.
    Let $y$ be the vertex succeeding $x$ on $P$.
    We have already argued that $xy\in E_t(x)$.

    We will show that $v\in \NL(y)$.
    Note that this will prove $v\in Y$ since $x\in B\cup Z^*$, $xy\in E_t(x)$ is an outgoing  edge to $V\setminus (B\cup Z^*)$, and consequently $\NL(y)\subseteq Y$ by the definition of $Y$.

    Since $y\notin B\cup Z^*$, the near-list $\NL(y)$ satisfies property~\ref{prop:dominate2} and $\NL(y)\cap U=\emptyset$.
    For contradiction, suppose $v\notin \NL(y)$.
    Consider the subpath $Q=y\to v$ of $P$.
    Observe that $Q$ is fully contained in $G-Z^*$, so $\dist_{G-Z^*}(y,v)=\dist_G(y,v)$.
    Since $Q$ is disjoint from $U$ and $y\notin U\cup \{s\}$, we in fact have $Q\subseteq G_0-Z$.
    As a result, $v$ is reachable from $y$ in $G_0-Z$.
    Since $v\notin \NL(y)$, by item~(a) of property~\ref{prop:dominate2}, we have
    $\NL(y)=\{(v_0,d_0),\ldots,(v_t,d_t)\}$, where $v_i\notin U$ and \[d_0,\ldots,d_t< \dist_{G_0-Z}(y,v)\leq \dist_{G-Z^*}(y,v).\]
    The last inequality follows from~\Cref {lem:contract}.
    By property~\ref{prop:path-length}, for any $u=0,\ldots,t$, $d_i$ corresponds to the length of some path $y\to v_i$ in $G_0$.
    Hence, there exists a path $s\to v_i$ in $G_0$ of length at most
    \begin{align*}
      \dist_G(s,x)+\wei(xy)+d_i&<\dist_G(s,x)+\wei(xy)+\dist_{G-Z^*}(y,v)\\
      &=\dist_G(s,x)+\wei(xy)+\dist_G(y,v)\\
      &=\dist_G(s,v).
    \end{align*}
    Since $G$ is obtained from $G_0$ by contraction that preserves distances from $s$ (see~\Cref{lem:contract}),
    we obtain ${\dist_G(s,v_i)<\dist_G(s,v)}$ for $t+1$ distinct vertices $v_i\in V$.
    This contradicts the assumption that $v\in \near{t}{s}$.
    Therefore, $v\in \NL(y)\subseteq Y$, as desired.
  \end{proof}
  Let us remark that, in general, neither the heavy set $Z^*$ alone, nor even $Z^*$ augmented with its out-neighbors' near-lists, always contains the next $t$ (or all remaining) undiscovered vertices.
  The additional subsets $B$ and $Y$ included in the subgraph $H$ in~\Cref{lem:equivalent-subgraph} are tailored specifically to address that problem, as demonstrated in the proof above.

\begin{observation}\label{obs:subgraph-size}
    $|Z^*\cup B\cup Y|=O(nt^2/p+\ell t^2p)$.
  \end{observation}
  \begin{proof}
    First, recall that $|Z|=O(nt/p)$ by~\Cref{lem:sparse-preprocessing} and we have $|Z^*|\leq |Z|$. Thus, $|Z^*|=O(nt/p)$.
    By property~\ref{prop:small-congest} of the near-lists (of $G_0$), for each $v\in U$, there exist $O(p)$ near-lists $\NL(u)$ such that $v\in\NL(u)$.
    Hence, the number of near-lists intersecting the set $U$ is $O(|U|p)=O(\ell tp)$.
    Equivalently, $|B|=O(\ell tp)$.
    Finally, in the set $Y$ we have at most $t+1$ vertices per outgoing edge $bv\in E_t(b)$ of each vertex $b\in B\cup Z^*$.
    Every vertex in $B\cup Z^*$ except $s$ has out-degree $O(1)$, whereas $|E_t(s)|$ is bounded by $t$.
    We conclude that $|Y|=O(|B\cup Z^*\setminus\{s\}|t+t^2)=O(nt^2/p+\ell t^2p)$.
  \end{proof}

  Equipped with~\Cref{lem:equivalent-subgraph} and~\Cref{obs:subgraph-size}, we are now ready to describe how each discovery step is performed precisely, assuming the preprocessing step has been performed at the beginning of the phase.
  We first set up the graph $H=G_t[Z^*\cup B\cup Y]$, as defined in~\Cref{lem:equivalent-subgraph}.
  Note that the vertices of $H$ depend on the discovery steps performed earlier on in the phase, or, in other words, on the set $U$ which grows while the phase proceeds.
  For example, in the first discovery step of a phase, we have $B=\emptyset$ and $Z^*=Z$ (recall that $Z$ includes $s$), whereas the set $Y$ contains the near-lists that can be ``reached'' via a single outgoing edge from $Z^*$.

  Constructing the sets $Z^*$, $B$, and $Y$ from the current $G$ and $U$ in near-optimal work and depth is possible assuming the preprocessed data is accompanied with ``inverse'' near-lists \linebreak $\NL^{-1}(v)=\{u\in V : v\in \NL(u)\}$.
  These can be easily set up after preprocessing in $\Ot(nt)$ work and $\Ot(1)$ depth.
  By~\Cref{obs:subgraph-size}, $H$ has $\Ot((n/p+\ell p)t^2)$ vertices.
  Since every vertex of~$H$ except the source $s$ has degree $O(1)$, we have $|E(H)|=O(|V(H)|)$.
  As a result, the graph $H$ can be constructed within $\Ot((n/p+\ell p)t^2)$ work and $\Ot(1)$ depth.

  From~\Cref{lem:equivalent-subgraph} it follows that the set $\near{t}{s}$ of $t$ closest vertices from $s$ is the same in the two graphs $G$ and $H$.
  Consequently, to find $\near{t}{s}$ in $G$, it is enough to apply~\Cref{lem:nearest-neighbors} to the subgraph~$H$
  instead of~$G$.
  This requires $\Ot((n/p+\ell p)t^4)$ work and $\Ot(1)$ depth.
  \subsection{Running time analysis}
  Let us summarize the total work and depth of the proposed SSSP algorithm that applies the near-list computation every $\ell$ discovery steps, and uses that to enable nearest vertices computation on a subgraph of $G$ as opposed to the entire graph $G$.
  
  Recall that due to our sparsity assumption, we have $m=O(n)$.
  We have already argued in~\Cref{sec:basic} that maintaining the adjacency lists (and even the top-$t$ adjacent edges of each vertex) subject to the performed contractions costs $\Ot(n)$ total work and $\Ot(n/t)$ depth.
  
  By~\Cref{lem:sparse-preprocessing}, the computation of the heavy vertices and near-lists required $\Ot(nt)$ work
  and $\Ot(t)$ depth.
  This preprocessing happens $O(n/(\ell t))$ times in sequence, so the total work spent on preprocessing is $\Ot(n^2/\ell)$, whereas the required depth is $\Ot(n/\ell)$.

  Finally, given the preprocessing, the total work consumed while performing the $O(n/t)$ discovery steps is $\Ot((n/t)\cdot (n/p+\ell p)t^4)=\Ot(n(n/p+\ell p)t^3)$.
  All the discovery steps run in near-optimal $\Ot(1)$ depth, so they are performed within $\Ot(n/t)$ depth in total.

  To sum up, the presented SSSP algorithm runs within $\Ot(n(n/p+\ell p)t^3+n^2/\ell)$ work and ${\Ot(n/t+n/\ell)}$ depth under the constant-degree assumption.

Let us now appropriately set the parameters $p$ and $\ell$ as a function of $n$ and the parameter $t$.
The constraints for the parameters are $p\in [2,n]$, $t\in [1,n]$ and $\ell\in [1,n/t]$.

  First of all, note that the choice $\ell<t$ is suboptimal: if this is the case, setting $t\coloneqq\ell$ decreases the work
  and does not increase the depth asymptotically.

  Note that the optimal choice for $p$ is such that $n/p=\ell p$, i.e., $p=\sqrt{n/\ell}$.
  In such a case, the work bound becomes $\Ot(n^{3/2}\ell^{1/2} t^3+n^2/\ell)$.
  
  Observe that if $t< \ell< n/t$ then the optimal choice for $\ell$ is $\ell=n^{1/3}/t^2$.
  Importantly, this is only valid for $t\in [1,n^{1/9})$ by $\ell>t$.
  With such an $\ell$, we conclude that $p$ should be set to $p=n^{1/3}t\in [1,n]$.
  So for $t\in [1,n^{1/9}]$, we can achieve work $\Ot(n^{5/3}t^2)$ and depth $\Ot(n/t)$.
 
  If $\ell=n/t$, then the work is $\Ot(n^2t^{5/2})$, so the obtained trade-off is no better than in the basic algorithm described in~\Cref{sec:basic}.
  
  If $\ell=t$, then the work is $\Ot(n^{3/2}t^{7/2}+n^2/t)$, which is $\Ot(n^{3/2}t^{7/2})$ for $t\geq n^{1/9}$.
  This bound is tighter than $\Ot(n^{5/3}t^2)$ only for $t<n^{1/9}$, but allows choosing $t\in [1,n^{1/2}]$, so that $\ell=t\leq n/t$ is satisfied.
  For such $\ell$, we infer $p=\sqrt{n/t}\in [1,n]$ as well.
  To sum up, if we aim at depth $\Ot(n/t)$, then for $t\in [1,n^{1/9})$ the work bound is $\Ot(n^{3/2}t^{7/2})$, whereas for $t\in [n^{1/9},n^{1/2}]$ it is
  $\Ot(n^{5/3}t^2)$.

  We can combine these cases into a single bound $\Ot(n^{3/2}t^{7/2}+n^{5/3}t^2)$ valid for the entire range $t\in [1,n^{1/2}]$.
  It is worth noting that this bound is only better than in the basic algorithm (\Cref{sec:basic}) for $t\leq n^{1/5}$.
  Finally, by dropping the $O(1)$-degree assumption, we obtain the bound ${\Ot(m^{3/2}t^{7/2}+m^{5/3}t^2)}$ valid for $t\in [1,m^{1/2}]$.

\sparsetradeoff*

\section{Lifting the algorithm to non-sparse graphs}\label{sec:dense}
The assumption about $G$ having constant degrees was crucial for the analysis of the algorithm presented in \Cref{sec:sparse}.
More specifically, it allowed us to efficiently bound the size of the subgraph which we computed $\near{t}{s}$ in (\Cref{obs:subgraph-size}) due to the following two reasons:
\begin{enumerate}
  \item We could bound $|B|$ by $O(\ell t p)$ since the number of occurrences of a single vertex in near-lists was at most~$p$ times the (constant) indegree.
	\item We had $|Y| \leq O(|B \cup Z^*|t)$ as the vertices of the graph (except the source) had constant outdegree.
\end{enumerate}

Our ultimate aim is to design an algorithm for general graphs with work $\Ot(n^{2-\epsilon} + m)$ and depth $\Ot(n^{1-\epsilon})$ for some $\epsilon > 0$. For these graphs, we cannot use the bounded-degree assumption anymore, so we need to work around both issues above.
The latter problem is easier to deal with: We can simply use the~bound $|E_t(v)| \le t$ rather than $|E_t(v)| = \Oh(1)$ and get a slightly worse, but still satisfactory, bound on $|Y|$.
However, the first issue turns out to be much more prominent. To handle it, we are going to partially reuse the precomputed objects between the phases of the algorithm (rather than discard all precomputed information after each phase).
More precisely, we will additionally compute sets of \emph{permanently heavy} vertices and \emph{alive edges} that will be maintained (and efficiently updated) between the phases. The permanently heavy vertices will exhibit weaker structural properties than the heavy vertices in the sparse variant of the algorithm. Because of that, at the beginning of each phase we will still need to additionally compute an~additional set of \emph{temporarily heavy} vertices.
Then, the set of heavy vertices -- understood as all permanently or temporarily heavy vertices -- will satisfy additional structural properties, which will eventually allow us to argue the effectiveness of the improved algorithm.

\subsection{Permanently heavy vertices and alive edges}
This subsection aims to define and state the properties of permanently heavy vertices and alive edges, explain how to compute these objects, and how to maintain them efficiently.
Recall that the graph $G$ satisfies Assumptions~\ref{ass:prefix} and \ref{ass:path-weights} (\Cref{sec:simplifying}).

\begin{lemma} \label{lem:pre-init}
  Let $t, p \le n$ be positive integers. Then, in $\Ot(m)$ work and $\Ot(n / t)$ depth we can compute a subset $Z_0 \subseteq V(G)$ and a subgraph $G_0$ of $G$ such that $V(G) = V(G_0)$ and:
\begin{enumerate}[(i)]
	
	\item\label{prop:outdeg-lb}For each $v \in V$, we have $\outdeg_{G_0}(v) \ge t$ or $N^{+}_{G_0}(v) \supseteq N^{+}_G(v) \setminus Z_0$.
	\item\label{prop:outdeg-ub} For each $v \in V$, we have $\outdeg_{G_0}(v) \le 3t$.
	\item\label{prop:indeg-ub} For each $v \in V$, we have $\indeg_{G_0}(v) \le p$.
	\item\label{prop:small-perm} The set $Z_0$ is of size $\Oh(nt / p)$.
	\item\label{prop:top-t} For all vertices $v, x, y \in V$, if $w(vx) < w(vy)$ and $y \in N^{+}_{G_0}(v)$, then $x \in N^{+}_{G_0}(v) \cup Z_0$.
\end{enumerate}
\end{lemma}

The vertices of $Z_0$ above will be called \emph{permanently heavy}, while the edges forming $G_0$ will be called \emph{alive}.
We remark that the last condition in~\Cref{lem:pre-init} states, in other words, that the alive edges going out from a~vertex~$v$ are its top-$k$ edges, for some $k$, except for edges pointing to $Z_0$ that may or may not be included in $N^{+}_{G_0}(v)$.

\begin{proof}[Proof of~\Cref{lem:pre-init}]
	We will proceed in multiple steps, where we start with $Z_0 = E(G_0) = \emptyset$. In each step, we will add some edges to $E(G_0)$, while hypothetically also enlarging $Z_0$. Throughout the algorithm, each edge will be considered for inclusion in $E(G_0)$ exactly once. For each vertex $v\in V$, we maintain a \emph{pending out-adjacency list} $P_v$ sorted by the edge weights kept in a parallel BST in the increasing order of weights. We initialize $P_v$ with $N^{+}_G(v)$ for each $v \in V$. We think of $P_v$ as a~queue of edges with tail $v$ that are yet to be considered. An~edge is removed from $P_v$ whenever it is extracted for consideration or whenever the head of $e$ is added to $Z_0$.
	
	In each step, if $N^{+}_{G_0}(v)$ contains less than $t$ vertices, we define $C_v$ to be top-$2t$ edges from $P_v$ (or all of them if $|P_v| < 2t$). Otherwise, we define $C_v = \emptyset$. By $C = \bigcup_{v \in V} C_v$ we denote the set of edges that we consider adding to $G_0$ in this step. Extending $E(G_0)$ by the entire set $C$ might make the property \ref{prop:indeg-ub} violated: it could be the case that the indegree of some vertex $u$ in $G_0$ will exceed $p$, that is, $\indeg_{G_0}(u) \le p < \indeg_{G_0 \cup C}(u)$. In this case, we choose arbitrary $p - \indeg_{G_0}(u)$ edges of $C$ with head $u$, add them to $G_0$, and declare $u$ as \emph{permanently heavy} (i.e., add $u$ to $Z_0$). We perform these steps for each such $u$; and for each such $u$, we also remove $u$ from all pending out-adjacency lists $P_v$ in which they occur. For all the remaining vertices, we add to $G_0$ all of their incoming edges from $C$. After that, we remove all edges of $C$ from pending out-adjacency lists where they occur to reflect the fact that they were already considered. This concludes the description of a~single step. Note that a~step can be implemented in $\Ot(1)$ depth. We repeat this process as long as there is any vertex with less than $t$ outgoing alive edges whose pending out-adjacency list is not empty yet.

	Let us analyze the resulting $Z_0$ and $G_0$ and prove that they satisfy the required properties.
	First, it follows from the stopping condition that all vertices have either at least~$t$ alive outgoing edges, or all of their edges have been considered, in which case all their outgoing edges whose heads do not end up in $Z_0$ are included in $E(G_0)$. Hence, item~\ref{prop:outdeg-lb} holds.

  Next, for each vertex $v$, in a single step we add to $E(G_0)$ at most $2t$ edges with tail $v$ (and in particular no edges if the outdegree of $v$ already exceeds $t$). Therefore, the outdegree of $v$ cannot exceed $3t$ after each step, which proves \ref{prop:outdeg-ub}.

  The property \ref{prop:indeg-ub} holds, because whenever we were to add too many incoming edges to a vertex~$u$, we specifically limit ourselves to choosing some $p- \indeg_{G_0}(u)$ edges of the new candidate edges to be added to $E(G_0)$, so that the indegree of such vertex in $G_0$ becomes exactly $p$. 
	
  As for each $z \in Z_0$ we have $\indeg_{G_0}(z) = p$ and $\sum_{v \in V} \indeg_{G_0}(v) = \sum_{v \in V} \outdeg_{G_0}(v) \le 3nt$, we conclude that $|Z_0| \le 3nt/p = \Oh(nt/p)$, asserting \ref{prop:small-perm}.

	The last property \ref{prop:top-t} follows from the fact that we consider edges outgoing from each vertex~$v$ in the increasing order of their weights. Each edge extracted from $P_v$ for consideration and whose head does not end up in $Z_0$ is included in $E(G_0)$.
	
	Let us now analyse the performance of this algorithm. Let us note that in each step, except for the last one, there had to be a vertex whose outdegree in $G_0$ was still less than $t$ after that step and not all of its outgoing edges were considered yet. It means that in that step, we processed its~$2t$ outgoing edges and more than $t$ of them did not get added to $E(G_0)$.
  By construction, for each such edge, its head was introduced to $Z_0$ in that step. Therefore, at least $t$ new vertices were introduced to $Z_0$ in this step. It follows that the number of steps is $\Oh(n / t)$, and as each step can be performed in $\Ot(1)$ depth. Thus, the total depth of this algorithm is $\Ot(n/t)$. As each edge is considered only once, the total work is naturally $\Ot(m)$.
\end{proof}

At the beginning of our improved shortest paths algorithm, we initialize $Z_0$ and $G_0$ with the sets computed by the algorithm of \Cref{lem:pre-init}.
We now show that these objects can be efficiently maintained under contractions of vertices into the source vertex $s$ (as defined in~\Cref{def:contract}).

\begin{lemma} \label{lem:pre-init-upd}
  Let $G$, $Z_0$ and $G_0$ be as in \cref{lem:pre-init}. It is possible to efficiently update $Z_0$ and~$G_0$ subject to a contraction of a subset of vertices $D$ into a distinguished source vertex $s$ (see \Cref{def:contract}). 
  The total work required for processing a sequence of $c \geq 0$ such operations is $\Ot(m)$ and the total depth is $\Ot(c + n/t)$.
\end{lemma}

\begin{proof}
	We remark that from the definition, all sets $D$ in a series of operations will be disjoint. Each vertex can enter $Z_0$ only once and can exit it only once too (when it is a part of $D$ in some update). 
	
	After the contraction of $D$ into $s$, the graph $G_0$ resulting from the contraction may violate the conditions \ref{prop:outdeg-lb}-\ref{prop:top-t} for two reasons. First, whenever we contract a set of vertices into the source $s$, we remove from $G_0$ all edges with head $s$, so the outdegrees of some vertices in $G_0$ may drop below the required threshold $t$. Moreover, $\outdeg_{G_0}(s)$ may become too large.
	The latter issue can be easily resolved: after each contraction, it suffices to restrict $N^{+}_{G_0}(s)$ to its top-$t$ edges. 
	
	For the former issue, we reapply the same process as we did in the initialization.
  That is, we repeat the step outlined in the proof of \Cref{lem:pre-init} as long as there exists a~vertex with less than~$t$ outgoing alive edges and with a non-empty pending out-adjacency list.
  We point out that in order to do this efficiently, we do not discard the pending out-adjacency lists $P_v$ after the initialization -- these lists persist between the discussed contraction operations. In particular, we remark that upon contraction of $D$, we remove from each $P_v$ all vertices belonging to $D$.
	
  The time analysis is analogous to that in the proof of \Cref{lem:pre-init}. Let us call a step \emph{concluding} if it was the last performed step after performing a contraction, or \emph{non-concluding} otherwise. Again, during each non-concluding step, at least $t$ vertices enter $Z_0$. Since throughout any series of operations, each vertex can enter and exit $Z_0$ at most once (and cannot reenter $Z_0$ afterwards), it follows that the total number of non-concluding steps across all $c$ operations is bounded by $\Oh(n/t)$. There is exactly one concluding step per update, so the number of concluding steps is $c$. Hence, the total number of steps is $\Oh(n/t) + c$. As each step can be performed in $\Ot(1)$ depth, the total depth of performing all updates is $\Ot(n/t + c)$. Since each edge is considered only a constant number of times (when extracted from some $P_v$ or when its head is contracted into~$s$), the total work throughout the entire sequence of operations is $\Ot(m)$.
\end{proof}

\subsection{Improved near-lists}
The outline of the improved algorithm remains the same as in the sparse case. We still split the algorithm into phases; each phase consists of $\ell$ steps; in each step we discover $t$ vertices closest to~$s$ and contract them into $s$.
Again, at the beginning of each phase, we run (in parallel) $|V|$ instances of parallel Dijkstra's algorithm to compute the near-lists and temporarily heavy vertices.
Unfortunately, with no sparsity assumption, the computation of near-lists in the original graph $G$ turns out to be too costly.
Instead, near-lists with slightly weaker guarantees are computed in the sparse subgraph $G_0$ of $G$ containing only the alive edges of~$G$.
The details follow below.

\newcommand{\Zinit}{Z_{\mathrm{init}}}

\begin{lemma}\label{lem:dense-preprocessing}
  Let $p\in [2,n]$ and $t \in [1, n]$ be integers and $J = (V, E)$ be a weighted directed graph such that $\outdeg_J(v) = O(t)$ and $\indeg_J(v) \le p$ for each $v \in V \setminus \{s\}$, where $s$ is a distinguished source with no incoming edges. Moreover, let $\Zinit \subseteq V(J)$ be a subset of vertices of $J$ containing~$s$.
  
  Then, in $\Ot(nt^2)$ work and $\Ot(t)$ depth we can compute a set $Z \supseteq \Zinit$ of heavy vertices of size $\Oh(|\Zinit| + nt/p)$ and a collection of near-lists $\NL(u)$, $u \in V$, satisfying the following properties:
\begin{enumerate}[(i)]
	\item\label{prop:dense-basic-nl} Each near-list $\NL(u)$ contains $u$ and at most $t$ other vertices.
	\item\label{prop:dense-small-congest} Each vertex appears in at most $p^2$ near-lists.
	\item\label{prop:dense-path-length} If $(v,d)\in \NL(u)$ (that is, $d$ is the distance accompanying $v$ in $\NL(u)$), then $d$ equals the weight
	of some $u\to v$ path in $J$.
	\item\label{prop:dense-dominate2} Fix $u\in V\setminus Z$. 
	If $\NL(u)=\{(v_0,d_0),\ldots,(v_k,d_k)\}$ and $0=d_0< d_1< \ldots< d_k$, then:
	\begin{enumerate}[(a)]
		\item \label{prop:dense-a} If a~vertex $v$ is reachable from $u$ in $J-Z$ and $v\notin \NL(u)$, then $k=t$ and ${d_t< \dist_{J-Z}(u,v)}$.
		\item \label{prop:dense-b} $d_i\leq \dist_{J-Z}(u,v_i)$ holds for all $i=0,\ldots,k$.
	\end{enumerate}
\end{enumerate}
\end{lemma}
\begin{proof}
	This lemma is an analogue of \cref{lem:sparse-preprocessing} and its proof is very similar, so we are just going to highlight the differences between them and explain how these differences influence the resulting statement of the lemma.
	
	The differences in the assumptions are that in \cref{lem:sparse-preprocessing}, the vertices were assumed to have constant degrees (except $s$), while here outdegrees and indegrees are upper-bounded by $O(t)$ and~$p$, respectively. Moreover, some subset $\Zinit$ of vertices is already initially classified as heavy.
	
  The preprocessing algorithm is the same as the one described in~\Cref{sec:sparse-prep}; however, its output and running time are different.
	
	The only difference stemming from the higher outdegrees is the fact that when adding a new vertex to some near-list, we have to process its at most $\Oh(t)$, rather than $\Oh(1)$, outgoing edges, which causes the work of a single Dijkstra run to grow from $\Ot(t)$ to $\Ot(t^2)$. One can also easily verify that the inclusion of a~vertex into a~near-list can still be processed in $\Ot(1)$ depth, so the total depth of a~single Dijkstra run remains $\Ot(t)$. The total work, however, increases from $\Ot(nt)$ to $\Ot(nt^2)$.

	The only place where the constant indegree assumption was used in the proof of \cref{lem:sparse-preprocessing} was to argue that the number of occurrences of any vertex $u$ on all near-lists cannot be larger than $(p-1) \cdot (\indeg(u)+1)$.
  As indegrees in $J$ are now bounded by $p$, the analogous upper bound in property~\ref{prop:dense-small-congest} becomes $p^2-1<p^2$ in our case.
	
	As for the initialization of heavy vertices with $\Zinit$, its size is just added to the final size of~$Z$. The bound on the number of heavy vertices added throughout the process (that is, $Z \setminus \Zinit$) remains $\Oh(nt/p)$ since the total size of all near-lists is at most $n(t+1)$ and each vertex from $Z \setminus \Zinit$ appears in at least $p$ near-lists.
\end{proof}

Note that the preconditions of \Cref{lem:dense-preprocessing} give restrictions on out- and indegrees of the vertices of the graph; therefore, we cannot apply this lemma to the graph $J \coloneqq G$, and instead we apply it to $J \coloneqq G_0$ with $\Zinit \coloneqq Z_0$ (i.e., $\Zinit$ equal to the set of permanently heavy vertices). In the following, let $\{\NL(v):v\in V\}$ be the collection of near-lists returned by \Cref{lem:dense-preprocessing} and let $Z \supseteq Z_0$ be the accompanying set of heavy vertices.

Unfortunately, the near-lists $\NL(\cdot)$ are only valid for the \emph{subgraph} $G_0$ of $G$ containing the alive edges of $G$.
This poses a~problem: After preprocessing, we want to perform $\ell$ discovery steps, with each step finding $t$ next vertices closest to the source $s$ in $G$ and contracting these vertices into~$s$. Following the proof strategy of \cref{thm:sparse-tradeoff}, we would like to give an analogue of \cref{lem:equivalent-subgraph} providing at each discovery step a small subgraph of the original graph $G$ that is guaranteed to contain all $t$ vertices that are currently closest to $s$ in $G$, along with the shortest paths corresponding to them.
However, \Cref{lem:dense-preprocessing} is applied to $G_0$ rather than $G$, so the computed near-lists are only valid with respect to $G_0$, and so \cref{lem:equivalent-subgraph} does not readily apply to $G$.
To alleviate this issue, we will now define \emph{improved near-lists} that will essentially extend the structural properties of near-lists to the original graph $G$, allowing us to apply an analogue of \cref{lem:equivalent-subgraph} to the original graph $G$.

For every vertex $u \in V$, we define the \emph{improved near-list} $\NL'(u)$ as follows.
If $u \in Z$, then $\NL'(u) \coloneqq \NL(u)$.
Otherwise, let \[
	C_u \coloneqq \NL(u) \,\cup\, \{(y,\, d_i + w(v_i y)) \,:\, (v_i, d_i) \in \NL(u),\, v_iy\text{ is an alive edge} \}
\] be the set of \emph{candidates} for the improved near-list of $u$, containing all vertices of the near-list of~$u$, as well as all vertices reachable from a~vertex of a~near-list of $u$ via a~single alive edge (i.e., an~edge of $G_0$).
Then $\NL'(u)$ is defined to contain the $(t+1)$ \emph{distinct} vertices of $C_u$ (or all of them if $C_u$ contains fewer than $t+1$ distinct vertices) with smallest associated distances (of the form either $d_i$ or $d_i+w(v_iy)$).
We remark that $u$ is always included in $\NL'(u)$ and that $|\NL'(u)| \le t+1$.
\begin{lemma} \label{lem:improved-properties}
	The improved near-lists $\NL'(u)$, $u\in V$, satisfy the following properties:
	
	\begin{enumerate}[(i)]
		\setcounter{enumi}{4}
		\item \label{prop:improved-congestion}
		Each vertex appears in at most $p^3+p^2$ improved near-lists.
		\item \label{prop:improved-paths}
			If $(v, d) \in \NL'(u)$, then $d$ equals the weight of some $u \to v$ path in $G$.
		\item \label{prop:improved-dominate}
		Fix $u \in V \setminus Z$. Let $\NL'(u) = \{(v'_0, d'_0), \ldots, (v'_k, d'_k)\}$ and $0=d'_0 < d'_1 < \ldots < d'_k$. Then,
		if a vertex $v$ is reachable from $u$ in $G - Z$ and $v \not\in \NL'(u)$, then $k=t$ and $d'_t < \dist_{G - Z}(u, v).$
	\end{enumerate}
\end{lemma}

Note that the statement above does not contain an analogue of item \ref{prop:dense-b} of the property \ref{prop:dense-dominate2} of \cref{lem:dense-preprocessing}, as it will not be required in the improved algorithm anymore.

\begin{proof}
	The property \ref{prop:improved-congestion} results from the fact that for a vertex $v$ to be included in $\NL'(u)$, either~$v$ must have been included in $\NL(u)$ or one of its in-neighbours in $G_0$ must have been included in $\NL(u)$. However, we have that $|\{v\} \cup N^{-}_{G_0}(v)| \le p+1$ and each of these vertices was included in at most $p^2$ near-lists, so we conclude that $v$ appears in at most $p^3+p^2$ improved near-lists $\NL'(u)$.
	
	The property \ref{prop:improved-paths} is clear from the construction.
	
	For the property \ref{prop:improved-dominate}, let us consider a vertex $v$ that is reachable from $u$ in $G - Z$, but $v \not\in \NL'(u)$. Let $\NL(u) = \{(v_0, d_0), \ldots, (v_z, d_z)\}$. First, let us observe that paths corresponding to elements of $\NL'(u)$ were chosen as the shortest paths to distinct vertices from a superset of paths corresponding to elements of $\NL(u)$. Therefore, $d'_i \le d_i$ for any $0 \le i \le z$.
	
	Let us consider the following case distinction.
	\begin{enumerate}[(1)]
		\item $\dist_{G - Z}(u, v) = \dist_{G_0 - Z}(u, v)$.
		
		\begin{enumerate}
			\item $v \in \NL(u)$.
			
			In this case, based on the part \ref{prop:dense-b} of property \ref{prop:dense-dominate2} of \cref{lem:dense-preprocessing}, we know that $\NL(u)$ contains a pair $(v, d)$, where $d \le \dist_{G_0 - Z}(u, v) = \dist_{G - Z}(u, v)$. Then, in the construction of $\NL'(u)$, $v$ was added to $C_u$ at a distance at most $d$, but $v$ was not chosen as corresponding to one of the smallest $(t+1)$ distances. This means that $k=t$ and $d'_t<d \le \dist_{G - Z}(u, v)$.
		
			\item $v \not\in \NL(u)$.
		
			In this case, based on the part \ref{prop:dense-a} of property \ref{prop:dense-dominate2} of \cref{lem:dense-preprocessing}, we obtain $|\NL(u)| = t+1$ and $d'_t \le d_t < \dist_{G_0 - Z}(u, v) = \dist_{G - Z}(u, v)$.
		\end{enumerate}
		
		\item $\dist_{G - Z}(u, v) < \dist_{G_0 - Z}(u, v)$.
		
		We note that as $G_0$ is a subgraph of $G$, it is the case that $\dist_{G - Z}(u, v) \leq \dist_{G_0 - Z}(u, v)$, so this case is the complement of the previous one. Let us consider a shortest path $P = u \to v$ in $G - Z$. As $\dist_{G - Z}(u, v) < \dist_{G_0 - Z}(u, v)$, it is not contained in $G_0 - Z$, so it has some first edge $xy$ that is not in $G_0$, that is, $xy$ is not alive. Note that as $xy$ is the first edge on this path that is not in $G_0$, we have that $\dist_{G_0 - Z}(u, x) = \dist_{G - Z}(u, x)$. We again consider two cases, this time dependent on whether $x \in \NL(u)$:
		
		\begin{enumerate}
			\item $x \in \NL(u)$.
			
			In this case, $x=v_i$ for some $0 \le i \le z$ and $d_i \le \dist_{G_0 - Z}(u, x)$ by the part \ref{prop:dense-b} of property \ref{prop:dense-dominate2} of \cref{lem:dense-preprocessing}. In the construction of $\NL'(u)$, we considered all paths formed by extending the path corresponding to $(x, d_i)$ by a single alive edge. Their lengths are $d_i + w(x y_j)$ for some $y_1, \ldots, y_c$.
        Since $xy$ is not alive and $y\notin Z_0$, by item~\ref{prop:outdeg-lb} of~\Cref{lem:pre-init}, we have $c\ge t$.
        Similarly, as $xy$ is not alive and $y\notin Z_0$, by item~\ref{prop:top-t} of~\Cref{lem:pre-init}, we have $w(x y_j) < w(xy)$ for each $1 \le j \le c$.
        Therefore, we have
        \begin{align*}
          d_i + w(x y_j)
          &< \dist_{G_0 - Z}(u, x) + w(xy)\\
          &= \dist_{G - Z}(u, x) + w(xy)\\
          &= \dist_{G - Z}(u, y)\\
          &\le \dist_{G - Z}(u, v).
        \end{align*}
        Hence, $d_i$ and $d_i + w(xy_j)$ for $1 \le j \le c$ form at least $t+1$ distances smaller than $\dist_{G - Z}(u, v)$ corresponding to $t+1$ distinct vertices, proving that $|\NL'(u)| = t+1$ and $d'_t < \dist_{G - Z}(u, v)$. 
			
			\item $x \not\in \NL(u)$.
			
			In this case, based on part \ref{prop:dense-a} of property \ref{prop:dense-dominate2} of \cref{lem:dense-preprocessing} we have that $|\NL(u)| = t+1$ and $d'_t \le d_t < \dist_{G_0 - Z}(u, x) = \dist_{G - Z}(u, x) \le \dist_{G - Z}(u, v)$.
		\end{enumerate}
	\end{enumerate}
	We conclude the proof by case exhaustion.
\end{proof}

We note that the computation of the improved near-lists can be done in $\Ot(nt^2)$ work and $\Ot(1)$ depth directly from the definition using parallel BSTs.

\subsection{Discovery phase}
Having established the properties of improved near-lists $\NL'(u)$ and explained how they are computed, we are ready to present the full algorithm. It proceeds similarly as in the sparse case (\Cref{sec:sparse}), but incorporating the permanently heavy vertices, alive edges and improved near-list into the picture. Recall that at the beginning of the whole procedure, we initialize the set of permanently heavy vertices $Z_0$ ($|Z_0| \leq \Oh(nt / p)$) and the subgraph of alive edges $G_0$ by applying \cref{lem:pre-init}. Afterwards, we proceed in $\ceil{\frac{n-1}{\ell t}}$ phases, where in each phase we perform $\ell$ discovery steps, each discovering $t$ new closest vertices to $s$ (or all of them, if there are less than~$t$ of them in the last step) after contracting the previously discovered vertices into $s$. At the beginning of each phase, we run the initialization phase of \cref{lem:dense-preprocessing} on $G_0$ and $Z_0$ to produce $Z$ of size $\Oh(nt/p)$ and the near-lists $\NL(u)$, and then we improve them to $\NL'(u)$. After computing the improved near-lists, we use them to perform the discovery steps essentially in the same way as in the sparse algorithm in \cref{lem:equivalent-subgraph}: We again define $E_t(v)$ as the top-$t$ (i.e., $t$ with smallest weights) outgoing edges of $v$ in $G$ (or all of them if $\outdeg_G(v)<t$) and the following auxiliary subsets of $V$:
\begin{itemize}
	\item $Z^*=Z\setminus U$ -- the remaining heavy vertices in $G$ (where $U$ is the set of already discovered vertices in this phase), with $s\in Z^*$,
	\item $B=\{u\in V : \NL'(u)\cap U\neq\emptyset\}$ -- the vertices whose improved near-list contains a vertex discovered since the beginning of the phase,
	\item $Y=\bigcup_{b\in B\cup Z^*}\bigcup_{bv\in E_t(b), v\in V\setminus (B\cup Z^*)} \NL'(v)$ -- the vertices appearing in improved near-lists
    $\NL'(v)$ such that (1) $v$ is among the top-$t$ out-neighbors in $G$ of a vertex in $B\cup Z^*$, (2) $v\notin B\cup Z^*$, that is, $v$ is light and $\NL'(v)$ contains no vertices discovered in the current phase.
\end{itemize}

As each vertex appears in at most $\Oh(p^3)$ improved near-lists $\NL'(u)$ (as opposed to $\Oh(p)$ in the sparse algorithm) and we bound $|E_t(v)|$ with $t$ for each $v \in V$ (compared to $\Oh(1)$ for each $v \neq s$ in~\Cref{sec:sparse}), we get the following looser bound on $|Z^* \cup B \cup Y|$:
\begin{observation}\label{obs:dense-subgraph-size}
	$|Z^*\cup B\cup Y|=\Oh(nt^3/p+\ell t^3p^3)$.
\end{observation}
\begin{proof}
	We proceed similarly as in the proof of \cref{obs:subgraph-size}. We have $|Z^*| \le |Z| = \Oh(nt/p)$. As each vertex appears in at most $\Oh(p^3)$ improved near-lists, we deduce $|B| = \Oh(|U|p^3) = \Oh(\ell tp^3)$. Consequently, $|Y| \le (t+1) \sum_{b \in B \cup Z^*} |E_t(b)| = \Oh(t^2(nt/p + \ell tp^3)) = \Oh(nt^3/p+\ell t^3p^3)$.
\end{proof}

As the next step, we claim the following.
\begin{lemma} \label{lem:dense-equivalent-subgraph}
	Let $G_t = (V, \bigcup_{v \in V} E_t(v))$ and $H = G_t[Z^* \cup B \cup Y]$. Let $u\in \near{t}{s}$. Then $u\in V(H)$ and $\dist_G(s,u)=\dist_H(s,u)$.
\end{lemma}
\begin{proof}
  This lemma is an analogue of \cref{lem:equivalent-subgraph} from the sparse case, and its proof proceeds exactly as the proof of~\Cref{lem:equivalent-subgraph}. The properties of near-lists from \cref{lem:sparse-preprocessing} that the proof uses are \ref{prop:path-length} and the part \ref{prop:sparse-a} of \ref{prop:dense-dominate2}.\footnote{In particular, the proof of \cref{lem:equivalent-subgraph} does not require the part \ref{prop:sparse-b} of \ref{prop:dense-dominate2}, whose analogue for improved near-lists is not proved by us in \cref{lem:improved-properties}.
  } As both of these are also satisfied for improved near-lists (see items \ref{prop:improved-paths} and \ref{prop:improved-dominate} of \cref{lem:dense-equivalent-subgraph}), the proof of \cref{lem:equivalent-subgraph} applies here verbatim.
\end{proof}

\Cref{lem:dense-equivalent-subgraph} enables us to efficiently perform a single discovery step in low depth using \cref{lem:nearest-neighbors} for $H = G_t[Z^* \cup B \cup Y]$. Based on \cref{obs:dense-subgraph-size}, we have $|V(H)| = \Oh(nt^3/p+\ell t^3p^3)$. Since each vertex of $H$ has outdegree at most $t$, we also have $|E(H)| \le t |V(H)|$, so applying \cref{lem:nearest-neighbors} to $H$ takes $\Ot(nt^5/p+\ell t^5p^3)$ work and $\Ot(1)$ depth. Similarly to \cref{lem:equivalent-subgraph}, the construction of $H$ itself can be performed within near-optimal $\Ot(|E(H)|)$ work and $\Ot(1)$ depth by precomputing ``reverse'' improved near-lists during preprocessing. After each discovery step, we perform the corresponding contraction into $s$ of $t$ newly found vertices closest to $s$.

This concludes the description and the proof of the correctness of the improved algorithm for non-sparse graphs.
It remains to summarize the efficiency of the algorithm.

\subsection{Running time analysis} \label{subsec:dense-time}
The initialization of permanently heavy vertices and alive edges takes $\Ot(m)$ work and $\Ot(n/t)$ depth according to \cref{lem:pre-init}. They are updated after each discovery step. Since there are at most $\floor{\frac{n-1}{t}}$ such updates, all these updates take $\Ot(m)$ work and $\Ot(n/t)$ depth in total, according to \cref{lem:pre-init-upd}. The contractions themselves take $\Ot(m)$ work and $\Ot(n/t)$ depth in total as well. 

At the beginning of each phase, we perform the computation of near-lists and improved near-lists. The computation of near-lists takes $\Ot(nt^2)$ work and $\Ot(t)$ depth, while their improvement takes $\Ot(nt^2)$ work and $\Ot(1)$ depth. Since there are $\Oh(n/(\ell t))$ phases, the computation of these objects contributes $\Ot(n^2t/\ell)$ to the total work and $\Ot(n/\ell)$ to the total depth.

We also perform $\Oh(n/t)$ discovery steps, where each of them takes $\Ot(nt^5/p+\ell t^5p^3)$ work and $\Ot(1)$ depth, hence in total they contribute $\Ot(n^2t^4/p + n \ell t^4 p^3)$ to the total work and $\Ot(n/t)$ to the total depth. 

Summing up, the total work required by our algorithm is $\Ot(m+n^2t^4/p + n \ell t^4 p^3 + n^2t/\ell)$, while its total depth is $\Ot(n/t + n/\ell)$.

First, let us note that it is suboptimal to choose $t > \ell$, because decreasing $t$ to $\ell$ also decreases the work and does not increase depth asymptotically. The work is minimized for a value of $p$ such that $n^2t^4/p = n \ell t^4 p^3$; hence $p = n^{1/4} / \ell^{1/4}$. Since $1 \le \ell \le n$, we note that $p$ defined this way will also satisfy $1 \le p \le n$, hence it is always a valid choice. 
Plugging $p$ into the expression bounding the work, it becomes $\Ot(m+n^{7/4}t^4 \ell^{1/4} + n^2 t / \ell)$.
This expression is minimized for a value of $\ell$ such that $n^{7/4}t^4 \ell^{1/4} = n^2 t / \ell$; hence we obtain $\ell = n^{1/5} / t^{12/5}$. We note that such a choice always leads to $\ell t \le n$, which is a required condition for the number of phases to be valid.
Since $t \leq \ell$, we have $t \leq n^{1/17}$.
Plugging $\ell$ into the expression for work, it becomes $\Ot(m+n^{9/5} t^{17/5})$, which shows that for any value of a parameter $1 \le t \le n^{1/17}$, we have an algorithm whose work is $\Ot(m+n^{9/5} t^{17/5})$ and whose depth is $\Ot(n / t)$.

\section{Shortest paths with exponential and lexicographic weights}
\label{sec:large-weights}

In this section, we will showcase an~application of \Cref{thm:main} to the graphs whose edge weights and path weights are inherently polynomially unbounded: single-source lexicographic bottleneck shortest paths and single-source binary shortest paths. Namely, we prove the following theorems.
\begin{theorem}
  \label{thm:lex-sssp}
  Let $G = (V, E)$ be a~weighted digraph with weights in $\N$.
  Single-source lexicographic bottleneck shortest paths in $G$ can be computed within $\Ot(m + n^{9/5}t^{21/5})$ work and $\Ot(n / t)$ depth for $t \in [1, n^{1/21}]$.
\end{theorem}

\begin{theorem}
  \label{thm:power2-sssp}
  Let $G = (V, E)$ be a~weighted digraph, where each edge $e$ has weight of the form $2^{w_e}$, $w_e \in \N$, and its weight is given on input as $w_e$.
  Single-source distances in $G$ can be computed within $\Ot(m + n^{9/5}t^{21/5})$ work and $\Ot(n / t)$ depth for $t \in [1, n^{1/21}]$.
\end{theorem}

It turns out that both problems can be solved efficiently via the following unified framework: First, we observe that all path weights appearing in the runtime of the algorithm can be manipulated and compared effectively -- they can be added to each other in the setting of this algorithm with additional $\Ot(t)$ work and $\Ot(1)$ depth overhead, and they can be compared with $\Ot(1)$ work and depth overhead.
Therefore, given this efficient implementation of path weight operations, we can simply invoke the algorithm of \Cref{thm:main} to solve any instance of binary or lexicographic bottleneck shortest paths efficiently.
The details follow below.

Suppose that $\M = (S, +, 0, \leq)$ is a~totally ordered commutative monoid (representing possible path weights), and $W \subseteq S$ is a finite set of permissible edge weights.
We say that $\M$ \emph{efficiently supports $W$} if there exists a~parallel data structure maintaining representations of elements of $\M$ such that:
\begin{itemize}
    \item a representation of $0$ can be constructed in work and depth $\Ot(1)$,
    \item given a representation of $a \in \M$ and $k$ elements $w_1, \ldots, w_k$ of $W$, the representation of $a+w_1+\ldots+w_k$ can be constructed in work $\Ot(k)$ and depth $\Ot(1)$, \emph{without} destroying the representation of $a$, and
    \item elements of $\M$ can be compared in work and depth $\Ot(1)$.
\end{itemize}
The data structure operates under the assumption that each stored element of $\M$ is a~sum of at most $n^{O(1)}$ elements of $W$; then, the $\Ot(\cdot)$ notation hides factors poly-logarithmic in $n$ and $|W|$.

Consider the case of single-source binary shortest paths.
Given an~$n$-vertex, $m$-edge input graph with each edge $e$ of weight $2^{w_e}$, we choose $\M_{\rm bin} = \N$ (i.e., path weights are non-negative integers) and $W_{\rm bin} = \{2^i \,\mid\, i \in \{0, 1, \ldots, m^2-1\}\}$ (i.e., edge weights are integer powers of two smaller than $2^{m^2}$).
The condition that edge weights are smaller than $2^{m^2}$ can be ensured in $\Ot(n + m)$ time -- without changing the shape of the single-source shortest paths tree -- by applying the following observation repeatedly: Suppose that $L, \delta \in \N$ are integers such that $\delta \geq m$ and each $w_e$ is either at most $L$ or at least $L + \delta$.
Then all input values $w_e$ greater than or equal to $L + \delta$ can be decreased by $\delta - (m - 1)$ without changing the shape of the structure of shortest paths.
Similarly, instances of lexicographic bottleneck shortest paths are modeled by choosing $\M_{\rm lex} = (S, \cup, \emptyset, \leq_{\rm lex})$, where $S$ is the family of finite multisets of non-negative integers and $\leq_{\rm lex}$ is the lexicographic comparison of two multisets.
Formally, we say that $A \leq_{\rm lex} B$ if one of the following holds: $A = \emptyset$, $\max A < \max B$, or $\max A = \max B$ and $A \setminus \{\max A\} \leq_{\rm lex} B \setminus \{\max B\}$.
We also choose $W_{\rm lex} = \{\{0\}, \{1\}, \ldots, \{m - 1\}\}$.
Input edge weights can be assumed to be smaller than $m$ after a~straightforward initial preprocessing of the graph.

We delay proving that $\M_{\rm bin}$ (and, respectively, $\M_{\rm lex}$) efficiently supports $W_{\rm bin}$ (resp., $W_{\rm lex}$) until later in this section.
Instead, we will first focus on the following result:

\begin{lemma}
    \label{lem:sssp-for-effective-path-weights}
    Suppose $\M$ efficiently supports $W$.
    Then a~single-source shortest path tree in a~directed graph with edge weights in $W$, where the path weights are evaluated in $\M$, can be computed within $\Ot(m + n^{9/5}t^{21/5})$ work and $\Ot(n/t)$ depth for $t \in [1, n^{1/21}]$.
\end{lemma}
\begin{proof}
	The key difference between this statement and the \cref{thm:main} is the fact that the arithmetic operations on $\M$ guaranteed by efficiently supporting $W$ are not as general as operations on $\R$ in the addition-comparison model. Namely, given $a, b \in \R$, the addition-comparison model assumes that we can compute $a+b$ in $\Oh(1)$ time, while given $c, d \in \M$, we are not provided with the operation of efficiently constructing the value $c+d$. However, this can be remedied if we assume that $d$ admits a~representation as a sum of a small number of elements of $W$. Hence, in this proof we will revisit all steps of the algorithm of \cref{thm:main} and verify that every value $a \in \M$ representing a path weight is obtained as a sum of: the already-constructed weight of another path and at most~$t$ other edge weights.
    Therefore, a~data structure efficiently supporting $W$ will incur only an $\Ot(t)$ multiplicative overhead to the total work of the algorithm of \cref{thm:main} (in particular, only to the parts of the algorithm where new path weights are constructed) and $\Ot(1)$ overhead to the total depth.
	
	Before we start analyzing the specific phases of the algorithm, we remark that even though we assume that edges of the input graph are representable as single elements of $W$, this stops being true after the operation of contraction of a set $U \subseteq V$ into the source $s$ as after that operation the edges outgoing of $s$ represent weights of some paths in the original graph. However, $s$ will be the only vertex with outgoing edges that are not representable as single elements of $W$ and there are no edges incoming to $s$, so whenever the algorithm appends an~edge to the end of a~path, it is never an edge going out of $s$ (unless the path is a single-vertex path $s$, which is not an issue, either: this path has weight $0$).
	
	First, we consider the computation of permanently heavy vertices and alive edges (\cref{lem:pre-init}). In this part, we never compute any additional path weights; it suffices to maintain the weights of the edges going out from each vertex sorted increasingly. Hence, this part takes $\Ot(m)$ work and $\Ot(n / t)$ depth.
	
	Then comes the computation of near-lists $\NL(u)$ of \cref{lem:dense-preprocessing} and their improvement to $\NL'(u)$ of \cref{lem:improved-properties}. The path weights that we construct in these parts are always of form $d_i + w(v_iy)$, where $d_i$ is an already-constructed path weight and $w(v_iy)$ is the weight of a single edge. Hence, we get the same work and depth guarantees, which are $\Ot(nt^2)$ for work in both parts, $\Ot(t)$ for depth of \cref{lem:dense-preprocessing}, and $\Ot(1)$ for depth of improving near-lists. As these precomputation phases are done at the beginning of each of $\Oh(n / (\ell t))$ phases, the total work required for all such precomputation stages will be $\Ot(n^2t / \ell)$, while their total depth will be $\Ot(n / \ell)$.
	
	The next step is the construction of the graphs $G_t$ and $H$ of \cref{lem:dense-equivalent-subgraph}. It does not refer to any path weights whatsoever, so their construction requires the same work and depth as previously, which are $\Ot(nt^4/p + \ell t^4p^3)$ and $\Ot(1)$, respectively.
	
	After computing $H$ comes the discovery step based on \cref{lem:nearest-neighbors} -- the only part of the algorithm where the difference between the computational models appears. In the algorithm of \cref{lem:nearest-neighbors}, we construct the weights of paths as $d^i(u, v) + d^i(v, y)$ for some vertices $u, v, y \in V$, where $d^i$ is the function representing the weight of the lightest path containing at most $2^i$ hops. However, whenever we use $d^i(a, b)$ in such an~expression, we are guaranteed that $b \in N_t^i(a)$, which is the set of $t+1$ closest vertices to $a$ (including $a$) when considering paths of up to $2^i$ hops. Because of that, it is obvious that the path certifying $d^i(a, b)$ cannot have more than $t$ hops. Moreover, in the expression $d^i(u, v) + d^i(v, y)$, if $u \neq v$, the path certifying $d^i(v, y)$ will not contain any edge going out from $s$ (we remark that we cannot argue the same thing for the path certifying $d^i(u, v)$), hence $d^i(v, y)$ is representable as a sum of at most $t$ elements of $W$ (it is straightforward to also explicitly maintain the $(\leq t)$ elements of $W$ whose sum is $d^i(v, y)$, without further overhead to work and depth) and the expression $d^i(u, v) + d^i(v, y)$ can be efficiently constructed in $\Ot(t)$ work and $\Ot(1)$ depth. As such, in our case we obtain the same $\Ot(1)$ depth guarantee as \cref{lem:nearest-neighbors}; and the work guarantee worsens by a factor of $\Ot(t)$: now it is $\Oh(|V(H)|t^3+|E(H)|t) =  \Ot(nt^6/p + \ell t^6p^3)$. Recall that there are $\Oh(n / t)$ discovery steps, so the total work incurred by them will be $\Ot(n^2t^5/p + n\ell t^5p^3)$
	
    The new edge weights constructed during the contraction process are of the form $\dist(s, x) + w(x, v)$, so they can be constructed in $\M$ with $\Ot(1)$ overhead for both work and depth, resulting in $\Ot(m)$ work and $\Ot(n / t)$ depth across the whole algorithm.
	
	As the last step, we perform the update of permanently heavy vertices and alive edges of \cref{lem:pre-init-upd}. Similarly to the case of their initial computation, this part retains the work and depth guarantees of the original algorithm, that is, $\Ot(m)$ work and $\Ot(n / t)$ depth across the entire process. 
	
  Summing up, the total work changes to $\Ot(m+n^2t^5/p + n\ell t^5p^3 + n^2t / \ell)$, while the total depth remains at $\Ot(n / t + n / \ell)$. Performing a similar running time analysis as in \cref{subsec:dense-time}, for ${p = n^{1/4}/\ell^{1/4}}$ and $\ell = n^{1/5}/t^{16/5}$ we obtain an algorithm whose work is $\Ot(n^{9/5}t^{21/5})$ and whose depth is $\Ot(n / t)$, for any choice of parameter $1 \le t \le n^{1/21}$.
\end{proof}

Having established \Cref{lem:sssp-for-effective-path-weights}, it remains to show that $\M_{\rm bin}$ (resp., $\M_{\rm lex}$) can efficiently support $W_{\rm bin}$ (resp., $W_{\rm lex}$).
In both cases, we will represent elements of $\M_{\rm bin}$ and $\M_{\rm lex}$ using logarithmic-depth persistent binary trees. When implemented correctly, this will allow us to compare two elements of the semigroup by performing a~single tree descent along the tree representation of the elements.
Formally, we will reduce the problem at hand to a~dynamic binary tree comparison data structure.
In what follows, we assume that each non-leaf vertex of a~binary tree may have one or two children (left and right).

We then use the following notion of a~\emph{label} of a~node in a~binary tree.
The label $W(x)$ of a~node $x$ in a~tree $T$ is a~word over $\{{\tt L}, {\tt R}\}$ naturally describing the path from the root of $T$ to $x$: Following this unique simple path, we write ${\tt L}$ when descending to the left child, and ${\tt R}$ when descending to the right child. The nodes of the binary trees that we will consider will be assigned colors from a fixed set $\Sigma$ (maintaining some auxiliary information about the elements of $\M$ that will facilitate efficient comparisons of these elements).

For two binary trees $T_1$ and $T_2$ by $\iso(T_1, T_2)$ let us denote the set of nodes $x \in T_1$ such that there exists a node $y \in T_2$ such that $W(x) = W(y)$ and the corresponding subtrees rooted at $x$ and $y$  
are isomorphic (with the isomorphism respecting colors in a~natural way). It is clear to see that $V(T_1) \setminus \iso(T_1, T_2)$ is a prefix of $T_1$, that is, it is either empty, or it is a connected subset of vertices of $T_1$ containing its root. Consequently, let us define $\diff(T_1, T_2) = T_1[V(T_1) \setminus \iso(T_1, T_2)]$. 
\begin{lemma}
    \label{lem:dynamic-tree-comparison}
    There exists a~deterministic parallel dynamic data structure maintaining a~family of nonempty binary trees of depth at most $b \in \N$ and with nodes colored with colors from a~fixed set $\Sigma$, supporting the following operations:
    \begin{itemize}
        \item Create and return a~one-node binary tree, in work and depth $\Ot(1)$.
        \item Given binary trees $T$ and $T^{+}$, create and return a unique tree $T'$ such that $\diff(T', T) = T^{+}$. The operation does \emph{not destroy} the original tree $T$, i.e., after the operation, the family contains both $T$ and $T'$.
        The operation is performed in work $\Ot(|T^{+}|)$ and depth $\Ot(b)$.
        \item Given two binary trees $T_1$, $T_2$, either determine that $T_1 = T_2$, or find the first node that belongs to both $T_1$ and $T_2$ that witnesses that $T_1 \neq T_2$.
        Formally, find the lexicographically smallest word $w$ over $\{{\tt L}, {\tt R}\}$ such that there exists a~node $x_1$ in $T_1$ and a~node $x_2$ in $T_2$, both with label $w$, such that either exactly one of $x_1$ and $x_2$ has a~left (right) child, or $x_1, x_2$ are nodes of different colors.
        The operation is performed in work and depth $\Ot(b)$.
    \end{itemize}
    The data structure allows multiple queries of the same type to be performed in parallel, i.e., any batch of single-node tree creation queries can be performed in depth $\Ot(1)$, and a~batch of prefix-replacement queries or tree comparisons can be performed in depth $\Ot(b)$.
\end{lemma}
\begin{proof}
    In the data structure, we will dynamically assign to each node of each tree an~integer identifier, so that any two nodes stored in the data structure are given the same identifier if and only if the subtrees of the respective trees rooted at these nodes are isomorphic (with the isomorphism naturally respecting colors).
    This approach resembles the parallel variant of the \emph{dictionary of basic factors} by Karp, Miller, and Rosenberg~\cite{KarpMR72}, implemented e.g.\ by Crochemore and Rytter~\cite{CrochemoreR91}.
    Here, each node $x$ is assigned a~triple $(c, {\rm id}_{\tt L}, {\rm id}_{\tt R})$, where $c$ is the color of $x$, ${\rm id}_{\tt L}$ is the identifier of the left child of $x$ (or $\bot$ if $x$ has no left child), and ${\rm id}_{\tt R}$ is the identifier of the right child of $x$ (or $\bot$ if $x$ has no right child).
    We maintain a~parallel dictionary ${\cal D}$ storing a~mapping $\mu$ from all such triples to unique identifiers.
    Then the identifier of $x$ is precisely $\mu(c, {\rm id}_{\tt L}, {\rm id}_{\tt R})$.
    We also maintain the number $I$, initially $0$, of unique identifiers assigned so far.

    Creating a~one-node tree with the only node $x$ of color $c$ requires us to verify if $\mu(c, \bot, \bot)$ already exists; in this case, the identifier of $x$ is precisely $\mu(c, \bot, \bot)$; otherwise, we define the identifier of $x$ to be $I$, define $\mu(c, \bot, \bot) \coloneqq I$ and increase $I$ by $1$.
    In case of multiple single-node tree creation queries, we group the trees by the color of the node and ensure that the nodes of the same color are assigned the same identifier.
    This can be done in parallel in depth $\Ot(1)$.

    For prefix-replacement queries, we stress that the required work complexity $\Ot(|T^{+}|)$ may be significantly smaller than $\Oh(|T'|)$, which is the size of the representation of the tree $T'$ we need to construct. This is made possible by the fact that the representations of $T$ and $T'$ are allowed to share nodes of $\iso(T, T')$, which are already allocated as part of $T$. Hence, we do not have to allocate these nodes again for the representation of $T'$, and it only suffices to properly set pointers from the nodes of $T'$ to $\iso(T, T')$. That method is precisely the \emph{node (path) copying method} of design of fully persistent data structures \cite{persistent}.
    
    Firstly, for each node $x \in T^{+}$, we determine its corresponding node $B(x) \in T$ with the property that $W(x) = W(B(x))$ (or $\perp$ if such $B(x)$ does not exist). This part can be performed in $\Ot(|T^{+}|)$ work and $\Ot(b)$ depth by performing synchronous descents in $T$ and $T^{+}$ in parallel for each $x \in T^{+}$. Afterwards, for each $x \in T^{+}$ that does not have a left child, we check if there exists $y \in T$ such that $W(y) = W(x){\tt L}$ (where $W(x){\tt L}$ denotes the label of $x$ concatenated with a single letter~\texttt{L}), and if yes, we set $y$ as the left child of $x$. We perform this check by verifying if $B(x)$ has a left child. We perform the analogous procedure for the right children as well.
    
    It is clear that the identifiers of nodes of $\iso(T', T)$ are the same as the identifiers of the corresponding nodes of $\iso(T, T')$, so the only nodes that require the assignment of identifiers are the nodes of $T^{+}$.
    We will compute these identifiers bottom-up.
    That is, for each $i = b, b - 1, \ldots, 1$ in order, we gather all nodes of $T^{+}$ at depth $i$, map each node to its assigned triple $(c, {\rm id}_{\tt L}, {\rm id}_{\tt R})$, and use ${\cal D}$ to determine if the mapping $\mu$ already defines the identifier $\mu(c, {\rm id}_{\tt L}, {\rm id}_{\tt R})$ of the node; if this is the case, the work for the given node is done.
    Then we allocate new consecutive identifiers, starting from $I$ upwards, to all nodes for which the identifier is still unknown.
    Again, by grouping the nodes with the same assigned triple, we ensure that the nodes with the same assigned triple are allocated the same identifier.
    Then we update the mapping $\mu$ for all assigned triples with freshly allocated identifiers.
    The entire procedure takes work $\Ot(|T^{+}|)$ and depth $\Ot(b)$. A~batch of updates can also be processed in the same depth by performing the procedure for all queries in the batch, in parallel, for each $i = b, b - 1, \ldots, 1$ in order.

    Now we can compare two trees $T_1, T_2$ in work and depth $O(b)$ by performing a~single descent down the trees: First, letting $x_1$ (resp., $x_2$) be the root of $T_1$ (resp., $T_2$), we verify that the identifiers of $x_1$ and $x_2$ are different (otherwise, $T_1 = T_2$).
    If this is the case, we progress $x_1$ and $x_2$ down $T_1$ and $T_2$, respectively, until $x_1$ and $x_2$ witness that $T_1 \neq T_2$ (i.e., the two nodes have different sets of children, or have different colors).
    While performing the descent, we replace both $x_1$ and $x_2$ with their respective left children (if their identifiers are different), and with their respective right children otherwise (in which case these are guaranteed to have different identifiers).
    It is straightforward to verify that this procedure returns a~witness that $T_1 \neq T_2$ in work and depth $\Ot(b)$; and naturally, multiple such tree comparisons can be performed fully in parallel.
\end{proof}

We now use \Cref{lem:dynamic-tree-comparison} to show that the discussed semigroups can support their respective edge weight sets.

\begin{lemma}
    $\M_{\rm lex}$ can efficiently support $W_{\rm lex} = \{\{0\}, \ldots, \{m - 1\}\}$.
\end{lemma}
\begin{proof}
    For convenience, assume that $m = 2^b$ for some $b \geq 1$.
    Let $\Sigma \coloneqq \N$.
    We assign to an~element $A$ of $\M_{\rm lex}$ (i.e., a~multiset of non-negative integers below $m$) a~rooted binary tree $T(A)$ defined as follows: First, create a full binary tree with $2^b$ leaves indexed $2^b - 1, 2^b - 2, \ldots, 0$ from left to right.
    The leaf of index $i \in \{0, \ldots, 2^b - 1\}$ receives the color equal to the number of occurrences of $i$ in~$A$.
    Non-leaf nodes receive color $0$.
    Eventually, all leaves of color $0$ are removed from $T(A)$, as well as non-leaf vertices whose subtrees do not contain any leaves of color other than~$0$.

    We initialize a~data structure of \Cref{lem:dynamic-tree-comparison} to maintain the family of trees $T(A)$ for every considered set $A$.
    It is trivial to see that $T(\emptyset)$ can be constructed in time and work $\Ot(1)$.
    Similarly, given an~element $A$ of $\M_{\rm lex}$ represented by the tree $T(A)$, it is straightforward to add $k$ elements $w_1, \ldots, w_k \in W$ to $A$ and create the tree representation of an~element $A' = A \cup \{w_1, \ldots, w_k\} \in \M_{\rm lex}$: Observe that $T(A')$ differs from $T(A)$ in a~prefix of size $O(kb)$, and this prefix can be determined in $O(kb)$ work and $O(b)$ depth.
    Therefore, $T(A')$ can be constructed from $T(A)$ in the data structure of \Cref{lem:dynamic-tree-comparison} in work $\Ot(kb)$ and depth $\Ot(b)$.
    Finally, given two elements $A_1, A_2$ of $\M_{\rm lex}$, the relation $A_1 \leq_{\rm lex} A_2$ can be determined by comparing the trees $T(A_1)$ and $T(A_2)$ by finding the first node belonging both to $T(A_1)$ and $T(A_2)$ that witnesses that $T(A_1) \neq T(A_2)$, and examining the label of this node and the set of children of this node in both $T(A_1)$ and $T(A_2)$.
    Thus the comparison can be performed in work and depth $\Ot(b)$.

    Since $b = \Ot(1)$ (recall that $b = \log m$), the claim follows.
\end{proof}

This proves \Cref{thm:lex-sssp}.
The analogous proof for $\M_{\rm bin}$ is slightly more technical.
\begin{lemma}
    $\M_{\rm bin}$ can efficiently support $W_{\rm bin} = \{2^0, \ldots, 2^{m^2 - 1}\}$.
\end{lemma}
\begin{proof}
    \newcommand{\Carry}{\mathsf{Carry}}
    \newcommand{\Sum}{\mathsf{Sum}}
    Choose $b \in \Ot(1)$ so that $2^b > m^2 \cdot n^{O(1)}$.
    Let also $\Sigma = \{0, 1\}$.
    Given an~integer $x \in [0, 2^{2^b} - 1]$ with binary representation $\langle x_{2^b-1}x_{2^b-2} \ldots x_0 \rangle$, we define the rooted binary tree $T(x)$ as $T(x, [0, 2^b - 1])$, where $T(x, [\ell, r])$ is defined recursively as follows:
    \begin{itemize}[nosep]
        \item If $x_\ell = x_{\ell + 1} = \ldots = x_r$, then $T(x, [\ell, r])$ is a~leaf node with color $x_\ell$.
        \item Otherwise, letting $p = \left\lfloor \frac{\ell + r}{2} \right\rfloor$, we choose $T(x, [\ell, r])$ to be a~non-leaf node with color $0$, the left child $T(x, [p + 1, r])$ and the right child $T(x, [\ell, p])$.
    \end{itemize}
    Observe that a~tree representation of an~integer with $k$ set bits has $\Ot(k)$ nodes; in particular, the tree representation of $0$ contains only $1$ node (the zero leaf).

    Now suppose we are given two integers $x, y \geq 0$ with $x + y < 2^{2^b}$ together with their tree representations $T(x)$, $T(y)$.
    We will present how to determine $T(x + y)$ in work $\Ot(|T(y)|)$ and depth $\Ot(1)$.
    Essentially, we reuse the classic Brent--Kung implementation of the \emph{carry-lookahead method} for binary addition~\cite{BrentK82}, exploiting the tree representation of input numbers to parallelize the computation.
    While the details are fairly straightforward, we include them below for completeness.

    Let $\Yc$ be the set of intervals $[\ell, r]$ for which $T(y)$ contains the node $T(y, [\ell, r])$.
    Define for $[\ell, r] \in \Yc$ and $c \in \{0, 1\}$ the value $\Carry(\ell, r, c)$ and the tree $\Sum(\ell, r, c)$:
    \[
        \begin{split}
            \Carry(\ell, r, c) &= \left\lfloor 2^{-(r - \ell + 1)} \cdot (\langle x_r x_{r-1} \ldots x_l \rangle + \langle y_r y_{r-1} \ldots y_\ell \rangle + c)\right\rfloor, \\
            \Sum(\ell, r, c) &= T\left( (\langle x_r x_{r-1} \ldots x_l \rangle + \langle y_r y_{r-1} \ldots y_\ell \rangle + c) \mod{2^{r - \ell + 1}},\ [0,\ r - \ell ]\right).
        \end{split}
    \]
    In other words, suppose that we have to sum together the parts of the binary representations of $x$ and $y$ between the $\ell$-th and the $r$-th binary digit, and additionally a~carry bit at the $\ell$-th bit if $c = 1$.
    Then $\Sum(\ell, r, c)$ is the tree representation of the resulting sum between the $\ell$-th and the $r$-th binary digit, while $\Carry(\ell, r, c)$ represents the carry bit transferred from the $r$-th bit of the result.
    Given $T_x \coloneqq T(x, [\ell, r])$ and $T_y \coloneqq T(y, [\ell, r])$, the value $\Carry(\ell, r, c)$ can be computed as follows:
    \begin{itemize}[nosep]
        \item If $T_x$ and $T_y$ are leaves of opposite colors, then $\Carry(\ell, r, c) = c$.
        \item Otherwise, if either $T_x$ or $T_y$ is a~leaf of some color $d \in \{0, 1\}$, then $\Carry(\ell, r, c) = d$.
        \item Otherwise, let $p = \left\lfloor \frac{\ell + r}{2} \right\rfloor$. Then $\Carry(\ell, r, c) = \Carry(p + 1, r, \Carry(\ell, p, c))$.
    \end{itemize}
    The values $\Carry(\ell, r, c)$ can be computed for all intervals $[\ell, r] \in \Yc$ and $c \in \{0, 1\}$ in a~bottom-up manner in work $\Ot(|T(y)|)$ and depth $O(b) = \Ot(1)$.
    Then, the tree $T(x + y) = \Sum(0, 2^b - 1, 0)$ is computed in the top-down manner:
    Suppose we need to construct $\Sum(\ell, r, c)$ for $0 \leq \ell \leq r < 2^b$ and $c \in \{0, 1\}$, given $T_x \coloneqq T(x, [\ell, r])$ and $T_y \coloneqq T(y, [\ell, r])$.
    \begin{itemize}[nosep]
        \item If at least one of the trees $T_x, T_y$ is a~leaf of color $c$, then $\Sum(\ell, r, c)$ is the other tree.
        \item Otherwise, if at least one of the trees $T_x, T_y$ is a~leaf of color $1 - c$ (without loss of generality, $T_y$), then $\Sum(\ell, r, c) = T(\langle x_r x_{r-1} \ldots x_\ell \rangle \pm 1,\ [0,\ r - \ell])$.
        It can be easily verified that $T^+ \coloneqq \diff(\Sum(\ell, r, c), T_x)$ has the shape of a~single vertical path with leaves attached to the path, and $T^+$ can be found efficiently in work and depth $O(b)$ by performing a~single recursive descent from the root of $T_x$. We omit the technical details here.
        \item Otherwise, let $p = \left\lfloor \frac{\ell + r}{2} \right\rfloor$ and $c' = \Carry(\ell, p, c)$. Compute $R = \Sum(\ell, p, c)$ and $L = \Sum(p + 1, r, c')$ in parallel. If $L, R$ are both leaves of the same color, say $d$, then $\Sum(\ell, r, c)$ is a~leaf of color $d$. Otherwise, $\Sum(\ell, r, c)$ is a~tree with left child $L$ and right child $R$.
    \end{itemize}
    This way, $T(x + y)$ can be constructed from $T(x)$ and $T(y)$ in work $\Ot(|T(y)|)$ and depth $\Ot(1)$.

    It remains to argue that $\M_{\rm bin}$ efficiently supports $W_{\rm bin}$.
    Naturally, zero-initializations and comparisons can be performed on the tree representations of the integers in work and depth $\Ot(1)$.
    Now suppose we are given a~tree representation $T(x)$ of an~integer $x < 2^{2^b}$, and additionally $k$ values $2^{i_1}, \ldots, 2^{i_k} \in W_{bin}$, and we are to compute $T(x + 2^{i_1} + \ldots + 2^{i_k})$.
    We first recursively find $T(2^{i_1} + \ldots + 2^{i_k})$: For each $j \in [k]$, the tree $T(2^{i_j})$ can be constructed in work and depth $\Ot(1)$.
    Also, given $T(2^{i_\ell} + \ldots + 2^{i_p})$ and $T(2^{i_{p+1}} + \ldots + 2^{i_r})$ for $1 \leq i \leq p < r \leq k$, the tree $T(2^{i_\ell} + \ldots + 2^{i_r})$ can be constructed in work $\Ot(r - \ell)$ and depth $\Ot(1)$ via the method sketched above.
    Thus, $T_y \coloneqq T(2^{i_1} + \ldots + 2^{i_k})$ can be constructed recursively, via a~standard halving divide-and-conquer approach, in work $\Ot(k)$ and depth $\Ot(1)$.
    The tree $T_y$ has just $\Ot(k)$ nodes; therefore, $T(x + 2^{i_1} + \ldots + 2^{i_k})$ can be finally constructed from $T(x)$ and $T_y$ in work $\Ot(k)$ and depth $\Ot(1)$.
\end{proof}

This, together with \Cref{lem:sssp-for-effective-path-weights}, concludes the proof of \Cref{thm:power2-sssp}.

\section{Incremental minimum cost-to-time ratio cycle}
\label{sec:mean-cycle}

In this section, we will use our improved shortest paths algorithm to design an~efficient data structure for the incremental variant of the generalization of the minimum-mean cycle problem, dubbed \textsc{Minimum Cost-to-Time Ratio Cycle}.
In the static version of the problem, we are given a~directed graph $G$ with $n$ vertices and $m$ edges, where each edge $e$ has an~arbitrary real cost $c(e)$ and positive real time $t(e)$.
The task is to determine a~cycle $C$ in $G$ with the minimum ratio of the total cost to the total time (further simply referred to as \emph{minimum-ratio cycle}), that is, the cycle minimizing the value $\frac{\sum_{e \in E(C)} c(e)}{\sum_{e \in E(C)} t(e)}$.
In the incremental variant, edges are dynamically inserted into the graph, and after each update, our aim is to report the current minimum-ratio cycle in $G$ (or correctly state that $G$ is acyclic).

Our result is as follows:

\begin{theorem}
    \label{thm:incremental-mean-cycle}
    There exists an~incremental data structure for \emph{\textsc{Minimum Cost-to-Time Ratio Cycle}} in an~$n$-vertex graph.
    The insertion of the $m$-th edge into the graph is processed in time $\Ot(mn^{21/22})$.
\end{theorem}

The problem was first analyzed in the static setting by Dantzig, Blattner, and Rao~\cite{DantzigBR67} and Lawler~\cite{Lawler1972}.
In this setting, the two best strongly polynomial algorithms were given by Bringmann, Hansen, and Krinninger~\cite{BringmannHK17}, who find a~minimum-ratio cycle in an~$n$-vertex, $m$\nobreakdash-edge graph in time $\Ot(m^{3/4} n^{3/2})$ or $n^3 / 2^{\Omega(\sqrt{\log n})}$.
Their algorithm reduces the minimum-ratio cycle problem to multiple instances of single-source shortest paths in graphs without negative cycles.
We observe that in the incremental setting, we can actually guarantee that all such instances contain only edges of non-negative weights.
This allows us to process edge insertions more efficiently than by recomputing the answer from scratch on every graph update.
In particular, this allows us to achieve $\Ot(mn^{1 - \epsilon})$ edge insertion time, which -- to the best of our knowledge -- has not been achieved before for this problem in the strongly polynomial model.

\paragraph{Approach.}
Let $\lambda \in \R$.
We say that a~potential function $\phi : V(G) \to \R$ \emph{attests} $\lambda$ in $G$ if $c(uv) - \lambda t(uv) - \phi(u) + \phi(v) \geq 0$ for every edge $uv \in E(G)$.
We use the fact that:
\begin{lemma}[{see~\cite[Lemma 2, 4]{BringmannHK17} and references therein}]
    There exists a~potential function attesting $\lambda$ in $G$ if and only if $\lambda$ does not exceed the minimum ratio of a~cycle in $G$.
\end{lemma}

The implementation of our data structure maintains, apart from the graph $G$ itself:
\begin{itemize}
    \item a~minimum-ratio cycle $C$ and its cost-to-time ratio $\lambda^\star$, and
    \item a~potential function $\phi : V(G) \to \R$ attesting $\lambda^\star$ in $G$.
\end{itemize}
(As a~special case, if $G$ is acyclic, we define that $C = \bot$, $\lambda^\star = +\infty$ and $\phi$ is undefined.)

Our edge insertion algorithm exploits the technique of \emph{parametric search}, introduced in a~work of Megiddo~\cite{Megiddo81} and already used successfully in the static variant of minimum-ratio cycle by Megiddo~\cite{Megiddo81} and Bringmann, Hansen, and Krinninger~\cite{BringmannHK17}.
This method allows us to design an~efficient strongly polynomial sequential algorithm determining the value of a~hidden real parameter $\lambda^\star$ (in our case, precisely the minimum ratio of a~cycle in a~graph), given two strongly polynomial algorithms -- a~low-depth parallel algorithm and an~efficient sequential algorithm -- \emph{comparing} a~given input parameter $\lambda$ with $\lambda^\star$.
Below we give the formal statement of the technique; an~intuitive explanation of its implementation can be found in \cite[Section 2.1]{BringmannHK17}.

\begin{theorem}[{\cite{Megiddo81, AgarwalST94}}]
    \label{thm:parametric-search}
    Let $\lambda^\star \in \R$ be a~hidden real value.
    Suppose that there exist two strongly polynomial algorithms that input a~real value $\lambda \in \R$ and test whether $\lambda \leq \lambda^\star$:
    \begin{itemize}
        \item a~parallel algorithm with work $W_p$ and depth $D_p$, and
        \item a~sequential algorithm with time complexity $T_s$.
    \end{itemize}
    We moreover require that in the parallel algorithm, all comparisons involving $\lambda$ test the sign of a~polynomial in $\lambda$ of constant degree.
    Then there exists a~sequential strongly polynomial algorithm computing $\lambda^\star$ in time $\Oh(W_p + D_pT_s \log W_p)$.
\end{theorem}

{
    \newcommand{\Cold}{C_{\rm old}}
    \newcommand{\Cnew}{C_{\rm new}}
    \newcommand{\lold}{\lambda_{\rm old}}
    \newcommand{\lnew}{\lambda_{\rm new}}
    \newcommand{\lpq}{\lambda_{pq}}
    \newcommand{\phiold}{\phi_{\rm old}}
    \newcommand{\phinew}{\phi_{\rm new}}
    Armed with~\Cref{thm:parametric-search}, we can implement the edge insertion.
    Suppose $G$ is a~directed graph just after an~insertion of the edge $pq$.
    Let $\Cold$ be a~minimum-ratio cycle in $G - pq$ (of ratio $\lold$) and $\phiold$ be a~potential function attesting $\lold$ in $G - pq$.
    Our aim is to determine a~minimum-ratio cycle $\Cnew$ in $G$ (of ratio $\lnew$) and a~potential function $\phinew$ attesting $\lnew$ in $G$.
    For our convenience, we assume that $m \geq n$, $\Cold \neq \bot$ and $\lold < +\infty$; it is straightforward to adapt the following description to the case where $m < n$ or $G - pq$ is acyclic.

    Observe that $\lnew = \min\{\lold, \lpq\} \leq \lold$, where $\lpq$ is the minimum-ratio cycle containing the edge $pq$ (where $\lpq = +\infty$ if no cycle in $G$ contains $pq$).
    Moreover,
    \[
        \lpq = \min \left\{
            \frac{c(pq) + \sum_{e \in E(P)} c(e)}{t(pq) + \sum_{e \in E(P)} t(e)} \,\colon
            P\text{ is a~path from $q$ to $p$ in $G - pq$}
        \right\}.
    \]
    In particular, for $\lambda \in \R$ we have $\lambda \leq \lpq$ if and only if for every path $P$ from $q$ to $p$ in $G - pq$ we have $\lambda \leq \frac{c(pq) + \sum_{e \in E(P)} c(e)}{t(pq) + \sum_{e \in E(P)} t(e)}$, or equivalently $\sum_{e \in E(P)} (c(e) - \lambda t(e)) \geq \lambda t(pq) - c(pq)$.

    Therefore, we can test whether $\lambda \leq \lnew$ for a~given $\lambda \in \R$ as follows.
    If $\lambda > \lold$, then also $\lambda > \lnew$ and the test is negative.
    So suppose that $\lambda \leq \lold$; then $\lambda \leq \lnew$ if and only if $\lambda \leq \lpq$.
    Construct an~edge-weighted graph $G_\lambda$ with the same set of edges as $G - pq$, where an~edge $uv \in E(G')$ is assigned weight $w_\lambda(uv) \coloneqq c(uv) - \lambda t(uv) - \phiold(u) + \phiold(v)$.
    Under our assumption, for every edge $uv \in E(G_\lambda)$, we have
    \[
        w_\lambda(uv) = c(uv) - \lambda t(uv) - \phiold(u) + \phiold(v) \geq c(uv) - \lold t(uv) - \phiold(u) + \phiold(v) \geq 0,
    \]
    since $\phiold$ attests $\lold$ in $G - pq$.
    Therefore, $G_\lambda$ has non-negative edge weights.
    Moreover, every path $P$ from $q$ to $p$ has total weight $\phiold(p) - \phiold(q) + \sum_{e \in E(P)} (c(e) - \lambda t(e))$, so $\lambda \leq \lpq$ if and only if every path from $q$ to $p$ in $G_\lambda$ has total weight at least $\Gamma_\lambda \coloneqq \lambda t(pq) - c(pq) + \phiold(p) - \phiold(q)$.
    Hence, in $O(m)$ work and $O(1)$ depth, testing whether $\lambda \leq \lnew$ can be reduced to a~single instance of non-negatively weighted single-source shortest paths in $G_\lambda$.
    Such an~instance can be solved by:
    \begin{itemize}[nosep]
        \item a~parallel algorithm with work $W_p = \Ot(m + n^{1 + 21/22})$ and depth $D_p = \Ot(n^{21/22})$ (\Cref{thm:main} for $t = n^{1/22}$), and
        \item a~sequential algorithm with running time $T_s = \Ot(n + m)$ (the usual Dijkstra's algorithm).
    \end{itemize}
    It can be verified that all edge and path weights in the parallel algorithm of \Cref{thm:main} are linear functions of $\lambda$.
    Thanks to that, \Cref{thm:parametric-search} applies and there exists a~sequential strongly polynomial algorithm computing $\lnew$ in time $\Oh(W_p + D_pT_s \log W_p) = \Ot(mn^{21/22})$.

    We now determine $\Cnew$.
    If $\lnew = \lold$, we can simply take $\Cnew = \Cold$.
    Otherwise, $\Cnew$ is a~cycle containing the new edge $pq$.
    In this case, the shortest path $P$ in $G_{\lnew}$ has total weight exactly $\Gamma_{\lnew}$ (and this path can be found by a~single run of Dijkstra's algorithm in $G_{\lnew}$), and then $\Cnew$ is constructed by concatenating $pq$ with $P$.

    It remains to update $\phiold$ to $\phinew$.
    Let $d : V(G) \to \R_{\geq 0} \cup \{+\infty\}$ denote the distance from $q$ to each vertex of $G_{\lnew}$ (where we set $d(v) = +\infty$ if $v$ is unreachable from $q$ in $G_{\lnew}$).
    For each $v \in V(G)$, we set $\phinew(v) = \phiold(v) - \delta(v)$, where
    \[
        \delta(v) \coloneqq \begin{cases}
            d(v) &\text{if } d(v) < +\infty, \\
            \Delta &\text{otherwise}
        \end{cases}
    \]
    for some large enough value $\Delta \geq 0$ (to be determined later).
    We verify now that $\phinew$ attests $\lnew$ in~$G$.
    Pick an~edge $uv \in E(G)$, aiming to show that
    \[ f(uv) \coloneqq c(uv) - \lnew t(uv) - \phinew(u) + \phinew(v) \geq 0. \]
    If $uv \neq pq$, then
    \[
        f(uv) \,=\, c(uv) - \lnew t(uv) - \phiold(u) + \phiold(v) + \delta(u) - \delta(v)
         \, =\, w_{\lnew}(uv) + \delta(u) - \delta(v).
    \]
    If $d(u) < +\infty$, then $d(v) < +\infty$ and $f(uv) = w_{\lnew}(uv) + d(u) - d(v) \geq 0$ since $d$ is a~distance function in $G_{\lnew}$.
    If $d(u) = d(v) = +\infty$, then $f(uv) = w_{\lnew}(uv) \geq 0$.
    Finally, if $d(u) = +\infty$ and $d(v) < +\infty$, then $f(uv) = w_{\lnew}(uv) + \Delta - d(v)$ and it is enough to ensure that $\Delta \geq d(v) - w_{\lnew}(uv)$.

    On the other hand, if $uv = pq$, then $\delta(q) = 0$ and
    \[
        f(pq) \,=\, c(pq) - \lnew t(pq) - \phiold(p) + \phiold(q) + \delta(p) - \delta(q)
         \, =\, \delta(p) - \Gamma_{\lnew}.
    \]
    If $d(p) < +\infty$, then $\delta(p) \geq \Gamma_{\lnew}$ follows from $\lnew \leq \lpq$.
    If $d(p) = +\infty$, then $f(pq) = \Delta - \Gamma_{\lnew}$ and it is enough to ensure that $\Delta \geq \Gamma_{\lnew}$.

    Hence $\phinew$ attests $G$ for
    \[ \Delta = \max\{0,\, \Gamma_{\lnew},\, \max\{d(v) - w_{\lnew}(uv) : uv \in E(G - pq),\, d(u) = +\infty,\, d(v) < +\infty\}\}. \]
    This concludes the edge insertion algorithm.
    We have shown it is correct and recomputes the minimum-ratio cycle and the potential function attesting the minimum ratio of a~cycle in $G$ in time $\Ot(mn^{21/22})$, thereby proving \Cref{thm:incremental-mean-cycle}.
}

\bibliographystyle{alpha}
\bibliography{references}

\clearpage
\appendix
\section{Simplifying assumptions}\label{sec:assumptions}

To see that Assumptions~\ref{ass:prefix}~and~\ref{ass:path-weights} can be always ensured in the addition-comparison~model with no loss of generality, interpret~$V$ as integers $\{0,\ldots,n-1\}$ and for each $uv=e\in E$, transform the weight of $e$ into the triple $w'(e):=(w(e),1,v-u)$
and add the tuples coordinate-wise when computing path weights.
Then, use the lexicographical order on triples to compare path lengths and distances from a single origin.
This way, the original edge weights are only manipulated via additions and comparisons, as desired.
On the other hand, manipulating the latter two coordinates of the augmented weights $w'$ of paths with $\poly(n)$ edges
only requires $O(1)$-time arithmetic operations on integers of magnitude $\poly(n)$; these are standard in the word RAM model. 

With this transformation, observe that~\Cref{ass:prefix} is guaranteed by non-negativity of $w$ and since the second coordinate of $w'(P')$ is strictly smaller than the second coordinate of $w'(P)$.
Similarly,~\Cref{ass:path-weights} is satisfied since $w'(P_1)$ and $w'(P_2)$, by a telescoping sum argument, have distinct third coordinates $u-x\neq u-y$.
If we optimize the distances lexicographically, the weights of shortest paths after the transformation have their respective first coordinates equal to the corresponding original distances in $G$, i.e., original distances can be retrieved trivially having solved the problem for transformed weights.

\section{Computing nearest vertices in low depth}\label{sec:nearest-neighbors}
\nearestneighbors*
\begin{proof}[Proof sketch]
  For $i=0,\ldots,\ell=\lceil \log_2{n}\rceil$, and $u\in V$, let $N_t^i(u)$ denote the closes $(t+1)$ (or all)
  vertices according to the constrained distance from $u$ \emph{via paths using at most $2^i$ edges}, \emph{including} $u$ itself.
  We will finally put $N_t(u) = N_t^{\ell}(u) \setminus \{u\}$, as $2^\ell\geq n$. Let us also define $d^i(u, v)$ to be the weight of the shortest path from $u$ to $v$ consisting of at most $2^i$ hops, if it exists, or $\infty$, if it does not. Observe that for any $u, v, y$ we have that $d^i(u, v) \ge d^{i+1}(u, v)$ and $d^{i+1}(u, y) \le d^i(u, v) + d^i(v, y)$. 
  
  Note that $N_t^0(u)$ for all $u\in V$ can be obtained within $\Ot(m)$ work and $\Ot(1)$ depth by identifying the $t$ lightest edges in the individual out-neighborhoods of the vertices in $G$.

  Suppose we have computed $N_t^i(u)$ for all $u\in V$, and the corresponding weights $d^i(u,v)$ of optimal paths from $u$ to $v\in N_t^i(u)$ using at most $2^i$ hops.
  We now explain how to compute all sets $N_t^{i+1}(\cdot)$ and the corresponding weights $d^{i+1}(\cdot,\cdot)$.
  Let us fix $u\in V$ and define \[ N'(u)=\bigcup_{v\in N_t^i(u)}N_t^i(v). \]
  For each $y\in N'(u)$, let \[ \delta(y) = \min\left\{d^i(u,v)+d^i(v,y) \,:\, v\in N_t^i(u),\, y\in N_t^i(v)\right\}.\]
  
  We prove the following claim:
  \begin{claim}
  $N_t^{i+1}(u)$ is the set consisting of $t+1$ vertices $y$ with the smallest values of $\delta(y)$ (or it is the whole $N'(u)$ if $|N'(u)| < t+1$). Moreover, the corresponding values $\delta(y)$ are equal to $d^{i+1}(u, y)$.
  \end{claim}
  \begin{claimproof}
  	Let us fix $y \in N_t^{i+1}(u)$. We aim to prove that $y \in N'(u)$ and that it is one of the $t+1$ vertices with the smallest value of $\delta$. Let $P$ be a shortest path from $u$ to $y$ of at most $2^{i+1}$ hops, and let $v$ be any vertex on $P$ such that its subpaths $P_1$ from $u$ to $v$ and $P_2$ from $v$ to $y$ consist of at most $2^i$ hops. We have \[ d^{i+1}(u, y) = w(P_1) + w(P_2) \ge d^i(u, v) + d^i(v, y) \ge d^{i+1}(u, y),\] so $w(P_1) = d^i(u, v)$ and $w(P_2) = d^i(v, y)$. Moreover, $w(P_1) \le w(P) = d^{i+1}(u, y)$ by non-negativity of weights, so $d^i(u, v) \le d^{i+1}(u, y)$. We now note that $v \in N_t^i(u)$, as otherwise there are at least $t+1$ distinct vertices $z\neq v$ such that $d^i(u, z) < d^i(u, v)$ (recall~\Cref{ass:path-weights}) which implies $d^{i+1}(u, z) \le d^i(u, z) < d^i(u, v) \le d^{i+1}(u, y)$ and thus $y \not\in N_t^{i+1}(u)$, a contradiction.
  	We can similarly show that $y \in N_t^i(v)$, as otherwise there exist at least $t+1$ vertices $z$ such that $d^i(v, z) < d^i(v, y)$ and for all such $z$ we have \[d^{i+1}(u, z) \le d^i(u, v) + d^i(v, z) < w(P_1) + d^i(v, y) = w(P_1) + w(P_2) = w(P)=d^{i+1}(u, y),\] a contradiction. The facts that $v \in N_t^i(u)$ and $y \in N_t^i(v)$ imply that $y \in N'(u)$ and that \linebreak $\delta(y) \le d^i(u, v) + d^i(v, y) = d^{i+1}(u, y)$. However, since for any vertices $a, b, c$ we have $d^{i+1}(a, c) \le d^i(a, b) + d^i(b, c)$, it is clear that $\delta(z) \ge d^{i+1}(u, z)$ for any $z \in N'(u)$. This implies $\delta(y) = d^{i+1}(u, y)$ and $\delta(y)$ indeed is among $t+1$ lowest values of~$\delta$, finishing the proof.
  \end{claimproof}
  
  The claim above shows that $N^{i+1}(u)$ can be constructed by computing $N'(u)$ and the corresponding values $\delta(y)$, and taking the $t+1$ lowest values of $\delta$. It is possible to do so for a particular $u \in V$ in $\Ot(t^2)$ work and $\Ot(1)$ depth using parallel sorting.
  Doing so for all vertices $u \in V$ in parallel takes $\Ot(nt^2)$ work and $\Ot(1)$ depth.
  We repeat this computation for subsequent $i=0, \ldots, \ell -1 = \ceil{\log n} -1$, achieving the goal of computing $N_t(u) \coloneqq N_t^{\ell}(u)$ along with the corresponding distances $\dist(u, v) \coloneqq d^{\ell}(u, v)$ within $\Ot(nt^2)$ work and $\Ot(\ell)=\Ot(1)$ depth.
\end{proof}

\end{document}